\documentclass[a4paper,UKenglish]{lipics-v2019}

\usepackage{microtype}

\usepackage[toc,page]{appendix}

\usepackage{booktabs}

\usepackage[utf8]{inputenc}

\usepackage{subcaption}
\usepackage{float}

\usepackage{tikz}
\usetikzlibrary{arrows,shapes,trees, arrows.meta}

\newtheorem{rules}{Rule}

\hideLIPIcs
\nolinenumbers

\makeatletter
\@fleqnfalse
\makeatother

\title{A General Stabilization Bound for Influence Propagation in Graphs}

\author{Pál András Papp}{ETH Zürich, Switzerland}{apapp@ethz.ch}{}{}
\author{Roger Wattenhofer}{ETH Zürich, Switzerland}{wattenhofer@ethz.ch}{}{}

\authorrunning{P.A. Papp and R. Wattenhofer}

\Copyright{Pál András Papp and Roger Wattenhofer}



\begin{document}

\maketitle

\begin{abstract}
We study the stabilization time of a wide class of processes on graphs, in which each node can only switch its state if it is motivated to do so by at least a $\frac{1+\lambda}{2}$ fraction of its neighbors, for some $0 < \lambda < 1$. Two examples of such processes are well-studied dynamically changing colorings in graphs: in majority processes, nodes switch to the most frequent color in their neighborhood, while in minority processes, nodes switch to the least frequent color in their neighborhood. We describe a non-elementary function $f(\lambda)$, and we show that in the sequential model, the worst-case stabilization time of these processes can completely be characterized by $f(\lambda)$. More precisely, we prove that for any $\epsilon>0$, $O(n^{1+f(\lambda)+\epsilon})$ is an upper bound on the stabilization time of any proportional majority/minority process, and we also show that there are graph constructions where stabilization indeed takes $\Omega(n^{1+f(\lambda)-\epsilon})$ steps.
\end{abstract}

\newpage
\setcounter{page}{1}

\section{Introduction}

Many natural phenomena can be modeled by graph processes, where each node of the graph is in a state (represented by a color), and each node can change its state based on the states of its neighbors. Such processes have been studied since the dawn of computer science, by, e.g., von Neumann, Ulam, and Conway. Among the numerous applications of these graph processes, the most eminent ones today are possibly neural networks, both biological and artificial.

Two fundamental graph processes are majority and minority processes. In a \textit{majority process}, each node wants to switch to the most frequent color in its neighborhood. Such a process is a straightforward model of influence spreading in networks, and as such, it has various applications in social science, political science, economics, and many more \cite{MajApplic1, MajApplic2, MajApplic3, MajApplic4, MajApplic5}.

In contrast, in a \textit{minority process}, each node wants to switch to the least frequent color in its neighborhood. Minority processes are used to model scenarios where the nodes are motivated to anti-coordinate with each other, like frequency selection in wireless communication, or differentiating from rival companies in economics \cite{KPRanticoor, MinApplic1, MinApplic2, MinApplic3, MinApplic4}.

Majority and minority processes have been studied in several different models, the most popular being the synchronous model (where in each step, all nodes can switch simultaneously) and the sequential model (where in each step, exactly one node switches). Since in many application areas, it is unrealistic to assume that nodes switch at the exact same time, we focus on the sequential model in this paper. We are interested in the worst-case stabilization time of such processes, i.e. the maximal number of steps until no node wants to change its color anymore.

Our main parameter describes how easily nodes will switch their color. Previously, the processes have mostly been studied under the basic switching rule, when nodes are willing switch their color for any small improvement. However, it is often more reasonable to assume a \textit{proportional switching rule}, i.e. that nodes only switch their color if they are motivated by at least, say, 70\% of their neighbors to do so. In general, we describe such proportional processes by a parameter $\lambda \in (0,1)$, and say that a node is switchable if it is in conflict with a $\frac{1+\lambda}{2}$ portion of its neighborhood. The stabilization time in such proportional processes (possibly as a function of $\lambda$) has so far remained unresolved.

The reason we can analyze proportional majority and minority processes together is that both can be viewed as a special case of a more general process of propagating conflicts through a network, where the cost of relaying conflicts through a node is proportional to the degree of the node. This more general process could also be used to model the propagation of information, energy, or some other entity through a network. This suggests that our results might also be useful for gaining insights into different processes in a wide range of other application areas, e.g. the behavior of neural networks.

In the paper, we provide a tight characterization of the maximal possible stabilization time of proportional majority and minority processes.
We show that for maximal stabilization, a critical parameter is the portion $\varphi$ of the neighborhood that nodes use as `outputs', i.e. neighbors they propagate conflicts to. 
Based on this, we prove that the stabilization time of proportional processes follows a transition between quadratic and linear time, described by the non-elementary function
\begin{equation} \label{eq:f}
f(\lambda) := \max_{\varphi \in (0,\frac{1-\lambda}{2}]} \; \frac{\log\left( \frac{1-\varphi}{\lambda+\varphi} \right)}{\log\left( \frac{1-\varphi}{\varphi} \right)}.
\end{equation}
More specifically, for any $\epsilon > 0$, we show that on the one hand, $O(n^{1+f(\lambda)+\epsilon})$ is an upper bound on the number of steps of any majority/minority process, and on the other hand, there indeed exists a graph construction where the processes last for $\Omega(n^{1+f(\lambda)-\epsilon})$ steps.

\section{Related Work}

Various aspects of both majority and minority processes on two colors have been studied extensively. This includes basic properties of the processes \cite{Goles, Winkler}, sets of critical nodes that dominate the process \cite{MajApplic3, MajOther1, Dynamo3}, complexity and approximability results \cite{votingtime, approx0, influenceApprox}, threshold behavior in random graphs \cite {MajOther2, Ahad2018}, and the analysis of stable states in the process \cite{SGPclass1, SGPclass3, SGPclass4, SGPsurvey, noUGP, KPRanticoor}. Modified process variants have also been studied \cite{certainityMaj, switchOnce}, with numerous generalizations aiming to provide a more realistic model for social networks \cite{SocialGen1, SocialGen2}.

However, the question of stabilization time in the processes has almost exclusively been studied for the basic switching rule (defined in Section \ref{sec:models}). Even for the basic rule, apart from a straightforward $O(n^2)$ upper bound, the question has remained open for a long time in case of both processes. It has recently been shown in \cite{majority} and \cite{minority} that both processes can exhibit almost-quadratic stabilization time in case of basic switching, both in the sequential adversarial and in the synchronous model.
On the other hand, the maximal stabilization time under proportional switching has remained open so far.

It has also been shown that if the order of nodes is chosen by a benevolent player, then the behavior of the two processes differs significantly, with the worst-case stabilization time being $O(n)$ for majority processes \cite{majority} and almost-quadratic for minority processes \cite{minority}. In weighted graphs, where the only available upper bound on stabilization time is exponential, it has been shown that both majority and minority can indeed last for an exponential number of steps in various models \cite{majorityW, minorityW}. The result of \cite{minorityW} is the only one to also study the proportional switching rule, showing that the exponential lower bound also holds in this case; however, since the paper studies weighted graphs with arbitrarily high weights, this model differs significantly from our unweighted setting.

Stabilization time has also been examined in several special cases, mostly assuming the synchronous model. 
The stabilization of a slightly different minority process variant (based on closed neighborhoods) has been studied in special classes of graphs including grids, trees and cycles \cite{CA1, CA2, CA3}. The work of \cite{hedetniemi} describes slightly modified versions of minority processes which may take $O(n^5)$ or $O(n^6)$ steps to stabilize, but provide better local minima (stable states) upon termination. For majority processes, stabilization has mostly been studied from a random initial coloring, on special classes of graphs such as grids, tori and expanders \cite{MajOther2, Ahad2018}.

Various aspects of majority processes have also been studied under the proportional switching rule, including sets of critical nodes that dominate the process, and sets of nodes that always preserve a specific color \cite{propDynamos1, propDynamos2}. However, to our knowledge, the stabilization time of the processes with proportional switching has not been studied before.

\section{Model and Notation}

\subsection{Preliminaries}

We define our processes on simple, unweighted, undirected graphs $G(V,E)$, with $V$ denoting the set of nodes and $E$ the set of edges. We denote the number of nodes by $n=|V|$. The neighborhood of $v$ is denoted by $N(v)$, the degree of $v$ by $\text{deg}(v)=|N(v)|$.

We also use simple directed graphs in our proofs. A directed graph is called a DAG if it contains no directed cycles. A \textit{dipartitioning} of a DAG is a disjoint partitioning ($V_1$, $V_2$) of $V$ such that each source node is in $V_1$, and all edges between $V_1$ and $V_2$ all go from $V_1$ to $V_2$. We refer to the set of edges from $V_1$ to $V_2$ as a \textit{dicut}.

Given an undirected graph $G$ with edge set $E$, we also define the \textit{directed edge set} of $G$ as $\widehat{E}=\{ (u,v), (v,u) \: | \: (u,v) \in E \}$, i.e. the set of directed edges obtained by taking each edge with both possible orientations.

A \textit{coloring} is a function $\gamma:V \rightarrow \{\text{black, white}\}$. A \textit{state} is a current coloring of $G$. Under a given coloring, we define $N_s(v)=\{ u \in N(v) | \gamma(v)=\gamma(u) \} $ and $N_o(v)=\{ u \in N(v) | \gamma(v) \neq \gamma(u) \} $ as the same-color and opposite-color neighborhood of $v$, respectively.

We say that there is a \textit{conflict} on edge $(u,v)$, or that $(u,v)$ is a \textit{conflicting edge}, if $u \in N_o(v)$ in case of a majority process, and if $u \in N_s(v)$ in case of a minority process. In general, we denote the conflict neighborhood by $N_c(v)$, meaning $N_c(v) = N_o(v)$ and $N_c(v) = N_s(v)$ in case of majority and minority processes, respectively. We occasionally also use $N_{\neg c}(v)=N(v) \setminus N_c(v)$.

If a node $v$ has more conflicts than a predefined threshold (depending on the so-called \textit{switching rule} in the model, discussed later) in the current state, then $v$ is \textit{switchable}. Switching $v$ changes its color to the opposite color. If edge $(u,v)$ becomes (ceases to be) a conflicting edge when node $v$ switches, then we say that $v$ has \textit{created} this conflict (\textit{removed} this conflict, respectively).

A \textit{majority/minority process} is a sequence of steps (states), where each state is obtained from the previous state by a set of switchable nodes switching. In this paper, we examine sequential processes, when in each step, exactly one node switches. Such a process is \textit{stable} when there are no more switchable nodes in the graph. By \textit{stabilization time}, we mean the number of steps until a stable state is reached.

\subsection{Model and switching rule} \label{sec:models}

We study the worst-case stabilization time of majority/minority processes, that is, the maximal number of steps achievable on any graph, from any initial coloring. In other words, we assume the \textit{sequential adversarial model}, when the order of nodes (i.e., the next switchable node to switch in each time step) is chosen by an adversary who maximizes stabilization time.

It only remains to specify the condition that allows a node to switch its color. The most straightforward switching rule is the following:

\begin{rules}[\textbf{Basic Switching}]
Node $v$ is switchable if $|N_c(v)| - |N_{\neg c}(v)| > 0$.
\end{rules}

An equivalent form of this rule is $|N_c(v)| > \frac{1}{2} \cdot \text{deg}(v)$. This rule is shown to allow up to $\widetilde{\Theta}(n^2)$ stabilization time for both majority \cite{majority} and minority \cite{minority} processes. However, it is often more realistic to assume a proportional switching rule, based on a real parameter $\lambda \in (0, 1)$:

\begin{rules}[\textbf{Proportional Switching}]
Node $v$ is switchable if $|N_c(v)| - |N_{\neg c}(v)| \geq \lambda \cdot \text{deg}(v)$.
\end{rules}

Since we have $|N_c(v)| + |N_{\neg c}(v)| = \text{deg}(v)$, this is equivalent to saying that $v$ is switchable exactly if $|N_c(v)| \geq \frac{1+\lambda}{2} \cdot \text{deg}(v)$. In the limit when $\lambda$ is infinitely small (or, equivalently, as $\frac{1+\lambda}{2}$ approaches $\frac{1}{2}$ from above), we obtain Rule I as a special case of Rule II.

In case of Rule I, whenever a node $v$ switches, it is possible that the total number of conflicts in the graph decreases by 1 only. On the other hand, Rule II implies that the switching of $v$ decreases the total number of conflicts at least by $\lambda \cdot \text{deg}(v)$ (we say that $v$ \textit{wastes} these conflicts), so in case of Rule II, the total number of conflicts can decrease more rapidly, allowing only a smaller stabilization time. Our findings show that the maximal number of steps is different for every distinct $\lambda$.

\subsection{On the $f(\lambda)$ function}

While the processes have a symmetric definition on each edge by default, it turns out that in order to maximize stabilization time, each edge has to be used in an asymmetric way. The most important parameter at each node $v$ is the ratio of neighbors $v$ uses as `inputs' and as `outputs'.
That is, the optimal behavior for each node $v$ is to select $\varphi \cdot \text{deg}(v)$ of its neighbors as outputs (for some $\varphi \in (0,1)$), and create all new conflicts on the edges leading to these output nodes, and similarly, mark the remaining $(1-\varphi) \cdot \text{deg}(v)$ neighbors as inputs, and only remove conflicts from the edges coming from these input nodes. Note that with Rule II, whenever a node switches, it can create at most $\left(1-\frac{1+\lambda}{2}\right) \cdot \text{deg}(v)=\frac{1-\lambda}{2} \cdot \text{deg}(v)$ new conflicts, so it is reasonable to assume $\varphi \in \left(0,\frac{1-\lambda}{2}\right]$.

Our results show that if all nodes select $\varphi$ as their output rate, then the maximal achievable stabilization time is a function of
\begin{equation} \label{eq:relay}
\frac{\log\left( \frac{1-\varphi}{\lambda+\varphi} \right)}{\log\left( \frac{1-\varphi}{\varphi} \right)}.
\end{equation}
As such, the largest stabilization time can be achieved by maximizing this expression by selecting the optimal $\varphi$ value, as shown in the definition of $f$ in Equation \ref{eq:f}. We denote the optimal value of $\varphi$ (i.e., the $\text{arg} \text{max}$ of Equation \ref{eq:relay}) by $\varphi^*$. The function $f$ has no straightforward closed form, as such a form would require solving 
\[  (\lambda+1) \cdot \varphi \cdot \log\left( \frac{1-\varphi}{\varphi} \right) = (\lambda+\varphi) \log \left( \frac{1-\varphi}{\lambda+\varphi} \right), \]
for $\varphi$, with $\lambda$ as a parameter. A more detailed discussion of $f$ is available in Appendix \ref{App:C}.

\begin{figure}
\captionsetup{justification=centering}
\centering
	\includegraphics[width=0.6\textwidth]{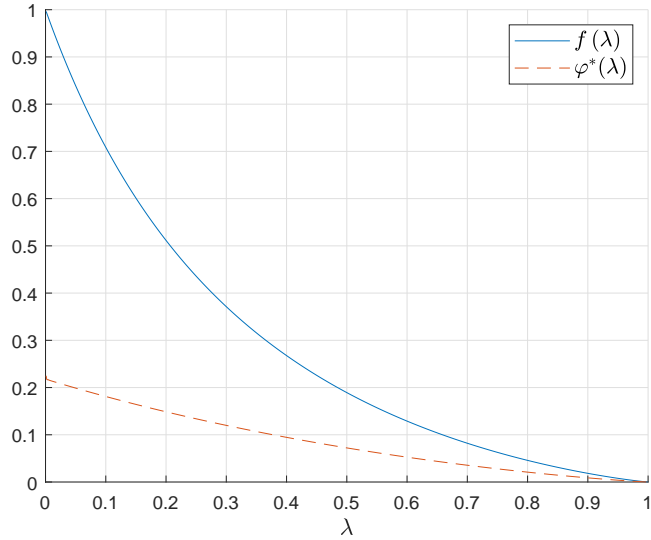}
	\caption{Plot of $f(\lambda)$ and $\varphi^*(\lambda)$ for $\lambda \in (0, 1)$}
	\label{fig:func}
\end{figure}

Figure \ref{fig:func} shows the values of $f$ and $\varphi^*$ as a function of $\lambda$. The figure shows that both $f(\lambda)$ and $\varphi^*(\lambda)$ are continuous, monotonically decreasing and convex.

It is visible that $\lim_{\lambda \rightarrow 0}f(\lambda)=1$ and $\lim_{\lambda \rightarrow 1}f(\lambda)=0$. This is in line with what we would expect: the simple switching rule allows a stabilization time up to $\widetilde{\Theta}(n^2)$ \cite{majority, minority}, while even for any large $\lambda<1$, it is still straightforward to present a graph with $\Omega(n)$ stabilization time. Our main result is showing that $f(\lambda)$ describes the continuous transition between these two extremes. 

\section{General intuition behind the proofs} \label{sec:prop}

Note that initially, each node $v$ can have at most $\text{deg}(v)$ conflicts on its incident edges, and each time when $v$ switches, it wastes $\lambda \cdot \text{deg}(v)$ conflicts. Therefore, if each node were to `use' its own initial conflicts only, then each node could switch at most $\frac{1}{\lambda}$ times, and stabilization time could never go above $O(n)$.

Instead, the idea is to take the high number of conflicts initially available at high-degree nodes, and use these conflicts to switch the less wasteful low-degree nodes many times. Specifically, we could have a set of $\Theta(n)$-degree nodes that initially have $\Omega(n^2)$ conflicts altogether on their incident edges, and somehow relay these conflicts to another set of $O(1)$-degree nodes, which only waste $O(1)$ conflicts at each switching. However, due to the large difference both in degree and in the number of switches, it is not possible to connect these two sets directly; instead, we need to do this through a range of intermediate levels, which exhibit decreasing degree and increasingly more switches. In order to maximize stabilization time, our main task is to move conflicts through these levels as efficiently (i.e., wasting as few conflicts in the process) as possible.

The formula of $f(\lambda)$ describes the efficiency of this process. The rate of inputs to outputs $\frac{1-\varphi}{\varphi}$ determines the factor by which the degree decreases at every new level. If $\varphi$ is chosen small, then $\frac{1-\varphi}{\varphi}$ is high, so we only have a few levels until we reach constant degree, and hence the number of switches is increased only a few times. On the other hand, the increase in the number of switches per level is expressed by $\frac{1-\varphi}{\lambda+\varphi}$, which is a decreasing function of $\varphi$. If $\varphi$ is too large, then although we execute this increase more times, each of these increases is significantly smaller.

With a degree decrease rate of $\frac{1-\varphi}{\varphi}$, we can altogether have about $\log_{\frac{1-\varphi}{\varphi}}(n)$ levels until the degree decreases from $\Theta(n)$ to $\Theta(1)$. If we increase the number of switches by a factor of $\frac{1-\varphi}{\lambda+\varphi}$ each time, then the $O(1)$-degree nodes will exhibit 
\begin{equation} \label{eq:match}
\left( \frac{1-\varphi}{\lambda+\varphi} \right)^{\log_{\frac{1-\varphi}{\varphi}}(n)} = n^{\frac{\log\left( \frac{1-\varphi}{\lambda+\varphi} \right)}{\log\left( \frac{1-\varphi}{\varphi} \right)}} \leq n^{f(\lambda)}
\end{equation}
switches, with an equation only if $\varphi=\varphi^*(\lambda)$. Having $\widetilde{\Theta}(n)$ nodes in the last level, this sums up to about $n^{1+f(\lambda)}$ switches altogether.

\subsection{Conflict propagation systems}

The upper bound on stabilization time is easiest to present in a general form that only focuses on this flow of conflicts in the graph. We define a simpler representation of the processes which only keeps a few necessary concepts to describe the flow of conflicts, and ignores e.g. the color of nodes or the timing of the switches at each node. In fact, we only require the number of times $s(v)$ each $v \in V$ switches, and the number $c(u,v)$ of conflicts that were created by node $u$ and then removed by node $v$, for each $(u,v) \in \widehat{E}$.

For simplicity, given a function $c:\widehat{E} \rightarrow \mathbb{N}$, let us introduce the notation $c_{in}(v) := \sum_{u \in N(v)} c(u,v)$ and $c_{out}(v) := \sum_{u \in N(v)} c(v,u)$.

\begin{definition}[\textbf{Conflict Propagation System, CPS}]
Given an undirected graph $G$, a conflict propagation system is an assignment $s:V \rightarrow \mathbb{N}$ and $c:\widehat{E} \rightarrow \mathbb{N}$ such that 
\begin{enumerate}
\item for each $v \in V$, we have $c_{in}(v) +\text{deg}(v) \geq \lambda \cdot \text{deg}(v) \cdot s(v) + c_{out}(v)$,
\item for each $v \in V$, we have $c_{out}(v) \leq \frac{1-\lambda}{2} \cdot \text{deg}(v) \cdot s(v)$, and
\item for each $(u,v) \in \widehat{E}$, we have $c(u,v) \leq s(u)$.
\end{enumerate}
\end{definition}

With the choice of $s(v)$ and $c(u,v)$ described above, any proportional majority or minority process indeed satisfies these properties, and thus provides a CPS. Hence if we upper bound the stabilization time (i.e. the total number of switches $\sum_{v \in V} s(v)$) of any CPS, this establishes the same bound on the stabilization time of any majority/minority process.

Condition 1 is the most complex of the three; it expresses the amount of `input conflicts' $c_{in}(v)$ required to switch $v$ an $s(v)$ times altogether. Every time after $v$ switches, it has at most $\frac{1-\lambda}{2} \cdot \text{deg}(v)$ conflicts on the incident edges, so it needs to acquire $\lambda \cdot \text{deg}(v)$ new conflicts to reach the threshold of $\frac{1+\lambda}{2} \cdot \text{deg}(v)$ and be switchable again; this results in the $\lambda \cdot \text{deg}(v) \cdot s(v)$ term. Moreover, if in the meantime, the neighboring nodes remove some of the conflicts from the incident edges (expressed by $c_{out}(v)$), then this also has to be compensated for by extra input conflicts. Finally, the extra $\text{deg}(v)$ term comes from the (at most) $\text{deg}(v)$ conflicts that are already on the incident edges in the initial coloring. For a detailed discussion of this condition, see Appendix \ref{App:A}.

Condition 2 also holds, since each time when $v$ switches, it creates at most $\frac{1-\lambda}{2} \cdot \text{deg}(v)$ conflicts on the incident edges. Each time $u$ switches, it can only create one conflict on a specific edge, so condition 3 also follows. Hence any majority/minority process indeed provides a CPS.

Finally, we need a technical step to get rid of the extra $\text{deg}(v)$ term in condition 1. Note that this term becomes asymptotically irrelevant as $s(v)$ grows; hence, our approach is to handle fewer-switching nodes separately, and require condition 1 only for nodes with large $s(v)$. More formally, we select a constant $s_0$, and we refer to nodes $v$ with $s(v) < s_0$ as \textit{base nodes}. We then consider \textit{Relaxed CPSs}, where, given this extra parameter $s_0$, condition 1 is replaced by:
\renewcommand{\labelenumi}{\arabic{enumi}R.}
\vspace{5pt}
\begin{enumerate}
{\setlength\itemindent{8pt} \item \textit{for each $v \in V$ with $s(v) \geq s_0$, we have $ c_{in}(v) \geq \lambda \cdot \text{deg}(v) \cdot s(v) + c_{out}(v)$,}}
\end{enumerate}
\vspace{6pt}
This relaxation comes at the cost of an extra $\epsilon$ additive term in the exponent of our upper bound.

\section{Upper bound proof} \label{sec:upper}

We now outline the proof of the upper bound on the number of switches. A more detailed discussion of this proof is available in Appendix \ref{App:A}.

\subsection{Properties of an optimal construction}

We start by noting that since moving a conflict through a node is wasteful, it is suboptimal to have two neighboring nodes that both transfer a conflict to each other, or more generally, to move a conflict along any directed cycle. Therefore, in a CPS with maximal stabilization time, the conflicts are essentially moved along the edges of a DAG. To formalize this, given a CPS, let us say that a directed edge $(u,v) \in \widehat{E}$ is a \textit{real edge} if $c(u,v)>0$.

\begin{lemma}
There exists a CPS with maximal stabilization time where the real edges form a DAG.
\end{lemma}

\begin{proof}
Among the CPSs on $n$ nodes with maximal stabilization time, let us take the CPS $P$ where the sum $\sum_{e \in \widehat{E}} c(e)$ is minimal. Assume that there is a directed cycle along the real edges of this CPS, and let $c(e_0)$ denote the minimal value of function $c$ along this cycle.

Now consider the CPS $P'$ where the value of $c$ on each edge of this directed cycle is decreased by $c(e_0)$. Since in each affected node, the inputs and outputs have been decreased by the same value, $P'$ still satisfies all three conditions, and thus it is also a valid CPS.  Moreover, $P'$ has the same amount of total switches as $P$. However, since $c(e_0)>0$, the sum of $c(e)$ values in $P'$ is less than in $P$, which contradicts the minimality of $P$. \qedhere

\end{proof}

Hence for the upper bound proof, we can assume that the real edges of the CPS form a DAG. In the rest of the section, we focus on this DAG composed of the real edges of the CPS. We first show that for convenience, we can also assume that each base node is a source in this DAG.

\begin{lemma} \label{baseSource}
There exists a CPS with maximal stabilization time where each base node is a source node of the DAG.
\end{lemma}

\begin{proof}
Note that by removing an input edge $(u,v)$ of a base node $v$ (that is, setting $c(u,v)$ to 0), the remaining CPS is still valid, since node $v$ does not have to satisfy condition 1R, and in node $u$, only the sum of outputs was decreased. Therefore, we can remove all the input edges of each base node, and hence base nodes will all become source nodes of the DAG. \qedhere
\end{proof}

\begin{lemma} \label{lem:constbase}
For each directed edge $(u,v)$ in the DAG where $u$ is a source node, $c(u,v)=O(1)$. More specifically, $c(u,v) \leq s_0$.
\end{lemma}

\begin{proof}
If $u$ is a base node, then $s(u) \leq s_0$, so $c(u,v) \leq s_0$ due to condition 3. Otherwise, condition 1R must hold, and since $u$ has no input nodes, we get $0 \geq c_{out}(u) + \lambda \cdot \text{deg}(u) \cdot s(u)$, hence $c_{out}(u)=0$, so $c(u,v) = 0$ for every $v$. Thus $c(u,v) \leq s_0$. \qedhere
\end{proof}

\subsection{Edge potential}

As a main ingredient of the proof, we define a way to measure how close we are to propagating conflicts optimally.
\begin{definition}[\textbf{Potential}]
Given a real edge $e \in \widehat{E}$,
the potential of $e$ is defined as $P(e)=c(e)^{1/f(\lambda)}$.
\end{definition}
\noindent For simplicity of notation, we also use $P$ to denote the function $x \rightarrow x^{1/f(\lambda)}$ on real numbers instead of edges.

Intuitively speaking, the potential function describes the cost of sending a specific number of conflicts through a single edge, in terms of the number of initial conflicts used up for this. Note that since $f(\lambda)<1$, the function $P$ is always convex. This shows that sending a high number of conflicts through a single edge is more costly than sending the same amount of conflicts through multiple edges.

As the following lemma shows, the potential is defined in such a way that the total potential can never increase when passing through a node in the DAG; the best that a node can do is to preserve the input potential if it relays conflicts optimally.

\begin{lemma} \label{preservePot}
For any non-source node $v$ of the DAG, with input edges from $N_{in}(v)$ and output edges to $N_{out}(v)$, we have
\[  \sum_{u \in N_{in}(v)} P(u,v) \geq \sum_{u \in N_{out}(v)} P(v,u) .\]
\end{lemma}

\begin{proof}
If $v$ is not a source, then by Lemma \ref{baseSource} it is not a base node, and thus has to satisfy condition 1R. In our DAG, $c_{in}$ and $c_{out}$ correspond to $\sum_{u \in N_{in}(v)} c(u,v)$ and $\sum_{u \in N_{out}(v)} c(v,u)$, respectively. Assume that we fix the value of $c_{in}$ and $c_{out}$. Since the potential function $P$ is convex, the incoming potential (left side) is minimized if $c_{in}$ is split as equally among the input neighbors as possible. On the other hand, the outgoing potential (right side) is maximized if $c_{out}$ is split as unequally among outputs as possible, so all output edges present in the DAG have the maximal possible number of switches, meaning $c(v,u)=s(v)$ for every $u \in N_{out}(v)$.

Assume that a fraction $\varphi$ of $v$'s incident edges are outgoing, i.e. $|N_{out}(v)|=\varphi \cdot \text{deg}(v)$ and $|N_{in}(v)|=(1-\varphi) \cdot \text{deg}(v)$. By condition 1R, we have $c_{in} \geq \lambda \cdot \text{deg}(v) \cdot s(v) + c_{out}$; with $c_{out}=\varphi \cdot \text{deg}(v) \cdot s(v)$, this gives $c_{in} \geq (\lambda+\varphi) \cdot \text{deg}(v) \cdot s(v)$. If split evenly among the $(1-\varphi) \cdot \text{deg}(v)$ inputs, this means
\[ \frac{c_{in}}{|N_{in}(v)|} \geq \frac{(\lambda+\varphi) \cdot \text{deg}(v) \cdot s(v)}{(1-\varphi) \cdot \text{deg}(v)} = \left( \frac{\lambda+\varphi}{1-\varphi} \right) \cdot s(v) \]
switches for each input node. The inequality on the potential then comes down to
\begin{gather*}
 \sum_{u \in N_{in}(v)} P(u,v) \geq (1-\varphi) \cdot \text{deg}(v) \cdot \left( \frac{\lambda+\varphi}{1-\varphi} \cdot s(v) \right)^{1/f(\lambda)} \geq \\ \geq \varphi \cdot \text{deg}(v) \cdot s(v)^{1/f(\lambda)} \geq \sum_{u \in N_{out}(v)} P(v,u).
\end{gather*}
To show that the inequality in the middle holds, we only require
\[ \left( \frac{\lambda+\varphi}{1-\varphi} \right)^{1/f(\lambda)} \geq \frac{\varphi}{1-\varphi}, \]
or, put otherwise,
\[ \frac{1}{f(\lambda)} \log \left( \frac{\lambda+\varphi}{1-\varphi} \right) \geq \log \left( \frac{\varphi}{1-\varphi} \right). \]
Since $\frac{\varphi}{1-\varphi}<1$ (thus its logarithm is negative), we get
\[  \frac{\log \left( \frac{\lambda+\varphi}{1-\varphi} \right)}{\log \left( \frac{\varphi}{1-\varphi} \right)} = \frac{\log \left( \frac{1-\varphi}{\lambda+\varphi} \right)}{\log \left( \frac{1-\varphi}{\varphi} \right)} \leq f(\lambda). \]
This holds by the definition of $f(\lambda)$. Note that this also shows that equality can only be achieved if the output rate $\varphi$ is indeed chosen as the argmax value $\varphi^*(\lambda)$. 
\end{proof}

Lemma \ref{preservePot} provides the key insight to the main idea of our proof: if we process the nodes of a DAG according to a topological ordering, always maintaining a dicut of outgoing edges from the already processed part of the DAG, then this potential cannot ever increase when adding a new node.

\begin{lemma} \label{cutbound}
Given a dicut $S$ of a dipartitioning in the DAG, we have
\[  \sum_{e \in S} P(e) = O(n^2) .\]
\end{lemma}

\renewcommand{\proofname}{Proof (Sketch).}

\begin{proof}

Each dipartitioning can be obtained by starting from the trivial dipartitioning where $V_1$ only contains the source nodes of the DAG, and then iteratively adding nodes one by one to this initial $V_1$. The number of outgoing edges from this initial $V_1$ (the set of source nodes) is upper bounded by $|E|=O(n^2)$. According to Lemma \ref{lem:constbase}, the number of switches (and hence the potential) on each edge of the dicut is at most constant, so the sum of potential in this initial dicut is also $O(n^2)$.

Now consider the process of iteratively adding nodes to this initial $V_1$ to obtain a specific dipartitioning. Whenever we add a new node $v$ to $V_1$, the incoming edges of $v$ are removed from the dicut, and the outgoing edges of $v$ are added to the dicut. According to Lemma \ref{preservePot}, the potential on the outgoing edges of $v$ is at most as much as the potential on the incoming edges, so the sum of potential can not increase in any of these steps. Therefore, when arriving at the final $V_1$, the sum of potential on the cut edges is still at most $O(n^2)$.
\end{proof}

\renewcommand{\proofname}{Proof}

\subsection{Upper bounding switches} \label{sec:upperlast}

Finally, we present our main lemma that uses the previous upper bound on potential in order to upper bound the number of switches in the CPS.

\begin{lemma} \label{lem:upper}
Given a CPS and an integer $a \in \{ 1, ..., n \}$, let $A=\{ v \in V \: | \: a \leq \text{deg(v) \textless \, 2a } \}$. For the total number of switches $s(A)=\sum_{v \in A} s(v)$, we have
\[ s(A) = O \left( n^{1+f(\lambda)} \cdot a^{-f(\lambda)} \right). \]
\end{lemma}

\renewcommand{\proofname}{Proof (Sketch).}

\begin{proof}

If the input edges of the nodes in $A$ would form the dicut of a dipartitioning, then we could directly use Lemma \ref{cutbound} to upper bound the number of switches in $A$ through the potential of the input edges. However, the nodes of $A$ might be scattered arbitrarily in the DAG, and if there is a directed path from one node in $A$ to another, then the `same' potential might be used to switch more than one node in $A$. Thus we cannot apply Lemma \ref{cutbound} directly. Instead, our proof consists of two parts.

1. First, we define so-called responsibilities for the nodes in $A$.
Given a node $v_0 \in A$, the idea is to devise two different functions: (i) a function $\Delta c(e)$, defined on each edge $e$ which is contained in any directed path starting from $v_0$, and (ii) a function $\Delta s(v)$, which is defined on any node $v$ that is reachable from $v_0$ on a directed path. Intuitively, we will consider the conflicts $\Delta c(e)$ and the switches $\Delta s(v)$ to be those that are indirectly `the effects of the switches of $v_0$'. More specifically, $\Delta c$ and $\Delta s$ are chosen such that if they are removed (subtracted from the CPS), then $v_0$ has no output edges in the DAG anymore, and the resulting assignment $s'(v)=s(v)-\Delta s(v)$ and $c'(e)=c(e)-\Delta c(e)$ still remains a valid CPS. Hence the subtraction results in a CPS where $v_0$ has no directed path to other nodes in $A$ anymore. This shows that we can keep on executing this step for each $v_0 \in A$ until no two nodes in $A$ are connected by a directed path, at which point we can apply Lemma \ref{cutbound} to the resulting graph.

Whenever we process such a node $v_0 \in A$, we define the \textit{responsibility} of $v_0$ as $R(v_0):=s(v_0) + \sum \Delta s (v)$, where the sum is understood over all the nodes $v \in A$ that are reachable from $v_0$. The main idea is that we `reassign' these switches to $v_0$ from other nodes in $A$. This method is essentially a redistribution of switches in the CPS, so we have $\sum_{v \in A}s(v) = \sum_{v \in A} R(v)$ altogether.

Furthermore, our definition of $\Delta s(v)$ will ensure that $R(v_0)=O(1) \cdot s(v_0)$. Intuitively, this can be explained as follows. Recall that with Rule II, the ratio of output to input conflicts is always upper bounded by a constant factor (below 1) at every node, since switching always wastes a specific proportion of conflicts. Hence, over any path starting from $v_0$, the number of outputs that can be attributed to $v_0$ forms a geometric series. As the ratio of the geometric series is below 1, the total amount of conflicts caused by $v_0$ this way is still within the magnitude of the input conflicts of $v_0$. Since each node in $A$ has similar degree (and thus requires similar number of input conflicts for one switching), these conflicts can only switch nodes in $A$ approximately the same number of times as $v_0$ can be switched by its own inputs. A more detailed discussion of this responsibility technique is available in Appendix \ref{App:A}.

2. For the second part of the proof, we show the claim in this modified CPS with no directed path between nodes in $A$. This implies that there exists a dipartitioning where the nodes of $A$ are in $V_2$, but all their input nodes are in $V_1$. This means that all the input edges of each node in $A$ are included in the dicut $S$ of the partitioning.

Consider a node $v \in A$. Due to condition 1R, $v$ has at least $\lambda \cdot \text{deg}(v) \cdot s(v)$ input conflicts. Even if these are distributed equally on all incident edges of $v$ (this is the case that amounts to the lowest total potential, since $P$ is convex), this requires a total input potential of
\[ \text{deg}(v) \cdot P(\lambda \cdot s(v)) = \text{deg}(v) \cdot s(v)^{1/f(\lambda)} \cdot \lambda^{1/f(\lambda)} \]
at least. Recall that Lemma \ref{cutbound} shows that the total potential on all edges in $S$ is $O(n^2)$. Our task is hence to find an upper bound on $\sum_{v \in A} s(v)$, subject to
\[ \sum_{v \in A} \text{deg}(v) \cdot s(v)^{1/f(\lambda)} \cdot \lambda^{1/f(\lambda)} = O(n^2) .\]
Since the last factor on the left side is a constant, we can simply remove it and include it in the $O(n^2)$ term. Furthermore, the degree of each node in $A$ is at least $a$, so by lower bounding each degree by $a$, we get
\[ \sum_{v \in A} s(v)^{1/f(\lambda)} = O(n^2) \cdot \frac{1}{a} .\]
Given this upper bound on $\sum_{v \in A} P(s(v))$, since the function $P$ is convex, the sum of switches $\sum_{v \in A} s(v)$ is maximal when each node in $A$ switches the same amount of times (i.e. there is an $s$ such that $s(v)=s$ for every $v \in A$), giving
\[ |A| \cdot s^{1/f(\lambda)} = O(n^2) \cdot \frac{1}{a} .\]
With this upper bound, $|A| \cdot s$ is maximal if $|A|$ is as large as possible and $s$ as small as possible (again because $P$ grows faster than linearly). Clearly $|A| \leq n$, so assuming $|A|=n$, we get
\[ s^{1/f(\lambda)} = O(n) \cdot \frac{1}{a} ,\]
which means that
\[ s = O(n^{f(\lambda)}) \cdot a^{-f(\lambda)} ,\]
and thus for the total number of switches in $A$, we get
\[ |A| \cdot s = O(n^{1+f(\lambda)}) \cdot a^{-f(\lambda)} . \qedhere \]

\end{proof}

\renewcommand{\proofname}{Proof}

It only remains to sum up this bound for the appropriate intervals to obtain our final bound. Let us consider the intervals $[1,2)$, $[2,4)$, $[4,8)$, ..., i.e. $a=2^k$ for each factor of 2 up to $n$, which is a disjoint partitioning of the possible degrees. Note that for these specific values of $a$, the sum $ \sum_{k=0}^{\infty} (2^k)^{-f(\lambda)}$ converges to a constant according to the ratio test. In other words, the sum is dominated by the number of switches of the lowest (constant) degree nodes, and hence, the total number of switches in the graph can be upper bounded by $O(1) \cdot n^{1+f(\lambda)}$.

Recall that since we work with Relaxed CPSs, we lose an $\epsilon$ in the exponent of this upper bound when we carry the result over to an original CPS.

\begin{theorem}
In any CPS with parameter $\lambda$, we have $\sum_{v \in V} s(v) =O(n^{1+f(\lambda)+\epsilon})$ for any $\epsilon>0$.
\end{theorem}

\noindent Since we have established that every majority/minority process provides a CPS, the upper bound on their stabilization time also follows.

\begin{corollary} \label{the:upper}
Under Rule II with any $\lambda \in (0,1)$, every majority/minority process stabilizes in time $O(n^{1+f(\lambda)+\epsilon})$ for any $\epsilon > 0$.
\end{corollary}

\section{Lower bound construction} \label{sec:lower}

Having established the most efficient way to relay conflicts, the high-level design of the matching lower bound construction is rather straightforward, following the level-based idea described in Section \ref{sec:prop}.

Given $\lambda$, we first determine the optimal output rate $\varphi=\varphi^*(\lambda)$. We then create a construction consisting of distinct levels, where each level has the same size, and each consists of a set of nodes that have the same degree. Since the degree should decrease by a factor of $\frac{\varphi}{1-\varphi}$ in each new level from top to bottom, we can add $L = \log_{\frac{1-\varphi}{\varphi}}(n)$ such levels to the graph. If each of these level has $\Theta(\frac{n}{\log{n}})$ nodes, then with the appropriate choice of constants, the total number of nodes is below $n$.

Each node in the construction is only connected to other nodes on the levels immediately above or below its own. All conflicts are propagated down in the graph, from upper to lower levels, so the upper neighbors of a node are always used as inputs, while the lower neighbors are always used as outputs. For the optimal propagation of conflicts, each node $v$ must have the optimal input-output rate, i.e. an up-degree of $(1-\varphi) \cdot \text{deg}(v)$ and a down-degree of $\varphi \cdot \text{deg}(v)$. Thus each consecutive level pair forms a regular bipartite graph, with $\frac{\varphi}{1-\varphi}$ of the degree of the level pair above. The construction is illustrated in Figure \ref{fig:constr}.

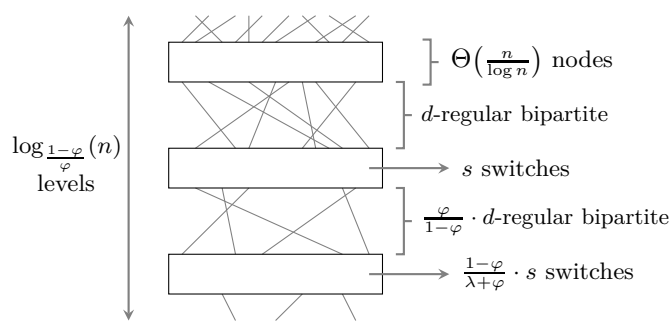
\begin{figure}
\centering
\captionsetup{justification=centering}
	\begin{tikzpicture}

	\draw[gray] (20pt,0pt) -- (25pt,-10pt);
	\draw[gray] (50pt,0pt) -- (40pt,-10pt);
	\draw[gray] (65pt,0pt) -- (70pt,-10pt);

	\draw[gray] (5pt,15pt) -- (30pt,40pt);
	\draw[gray] (25pt,15pt) -- (20pt,40pt);
	\draw[gray] (75pt,15pt) -- (65pt,40pt);
	\draw[gray] (35pt,15pt) -- (70pt,40pt);
	\draw[gray] (65pt,15pt) -- (10pt,40pt);
	
	\draw[gray] (10pt,55pt) -- (45pt,80pt);
	\draw[gray] (25pt,55pt) -- (5pt,80pt);
	\draw[gray] (30pt,55pt) -- (40pt,80pt);
	\draw[gray] (55pt,55pt) -- (50pt,80pt);
	\draw[gray] (60pt,55pt) -- (15pt,80pt);
	\draw[gray] (75pt,55pt) -- (55pt,80pt);
	\draw[gray] (50pt,55pt) -- (70pt,80pt);
	\draw[gray] (65pt,55pt) -- (30pt,80pt);
	
	\draw[gray] (5pt,95pt) -- (15pt,105pt);
	\draw[gray] (10pt,95pt) -- (25pt,105pt);
	\draw[gray] (20pt,95pt) -- (10pt,105pt);
	\draw[gray] (25pt,95pt) -- (35pt,105pt);
	\draw[gray] (30pt,95pt) -- (30pt,105pt);
	\draw[gray] (35pt,95pt) -- (45pt,105pt);
	\draw[gray] (40pt,95pt) -- (55pt,105pt);
	\draw[gray] (55pt,95pt) -- (70pt,105pt);
	\draw[gray] (60pt,95pt) -- (50pt,105pt);
	\draw[gray] (65pt,95pt) -- (75pt,105pt);
	\draw[gray] (75pt,95pt) -- (60pt,105pt);
	
	\draw (0pt,0pt) -- (0pt,15pt) -- (80pt,15pt) -- (80pt,0pt) -- cycle;
	\draw (0pt,40pt) -- (0pt,55pt) -- (80pt,55pt) -- (80pt,40pt) -- cycle;
	\draw (0pt,80pt) -- (0pt,95pt) -- (80pt,95pt) -- (80pt,80pt) -- cycle;
	
	\draw[gray, thick, arrows=stealth-stealth] (-15pt,-10pt) -- (-15pt,105pt);
	\node[anchor=east] at (-14pt,53pt) {\small $\log_{\frac{1-\varphi}{\varphi}}(n)$};
	\node[anchor=east] at (-24pt,42pt) {\small levels};
	
	\draw[gray, thick] (95pt, 96pt) -- (98pt, 96pt) -- (98pt, 79pt) -- (95pt, 79pt);
	\draw[gray, thick] (98pt, 87.5pt) -- (101pt, 87.5pt);
	\node[anchor=west] at (102pt,87.5pt) { $\Theta$\small$\left( \frac{n}{\log{n}} \right) $ nodes};

	\draw[gray, thick, arrows=-stealth] (75pt, 47.5pt) -- (105pt, 47.5pt);
	\node[anchor=west] at (106pt,47.5pt) {\small $s$ switches};
	
	\draw[gray, thick, arrows=-stealth] (75pt, 7.5pt) -- (105pt, 7.5pt);
	\node[anchor=west] at (106pt,7.5pt) {\small $\frac{1-\varphi}{\lambda+\varphi} \cdot s$ switches};
	
	\draw[gray, thick] (85pt, 80pt) -- (88pt, 80pt) -- (88pt, 55pt) -- (85pt, 55pt);
	\draw[gray, thick] (88pt, 67.5pt) -- (91pt, 67.5pt);
	\node[anchor=west] at (91pt,67.5pt) {\footnotesize $d$-regular bipartite};	
	
	\draw[gray, thick] (85pt, 40pt) -- (88pt, 40pt) -- (88pt, 15pt) -- (85pt, 15pt);
	\draw[gray, thick] (88pt, 27.5pt) -- (91pt, 27.5pt);
	\node[anchor=west] at (91pt,27.5pt) {\footnotesize $\frac{\varphi}{1-\varphi} \cdot d$-regular bipartite};	

\end{tikzpicture}
	\caption{Consecutive levels of the lower bound construction}
	\label{fig:constr}
\end{figure}

Our parameters $\lambda$ and $\varphi$ also determine that the number of switches should increase by a factor $\frac{1-\varphi}{\lambda+\varphi}$ on each new level. If we can always increase the switches at this rate, then each node on the lowermost level will switch
\begin{equation*} \label{eq:match2}
\left( \frac{1-\varphi}{\lambda+\varphi} \right)^{\log_{\frac{1-\varphi}{\varphi}}(n)} = n^{\frac{\log\left( \frac{1-\varphi}{\lambda+\varphi} \right)}{\log\left( \frac{1-\varphi}{\varphi} \right)}} = n^{f(\lambda)},
\end{equation*}
times, where the last equation holds because we are using $\varphi=\varphi^*(\lambda)$. Since there are $\widetilde{\Theta}(n)$ nodes on the lowermost level, the switches in this level already amount to a total of $\widetilde{\Theta}(n^{1+f(\lambda)})$, matching the upper bound.

However, note that when $\varphi^*(\lambda)$ or $\frac{1-\varphi}{\lambda+\varphi}$ is irrational, we can only use close enough rational approximations of these values. This comes at the cost of losing a small $\epsilon$ in the exponent.

\begin{theorem} \label{the:lower}
Under Rule II with a wide range of $\lambda$ values, there is a graph construction and initial coloring where majority/minority processes stabilize in time $\Omega(n^{1+f(\lambda)-\epsilon})$ for any $\epsilon > 0$.
\end{theorem}

This level-based structure describes the general idea behind our lower bound construction. However, the main challenge of the construction is in fact designing the connection between subsequent levels. In particular, this connection has to make sure that conflicts are indeed always relayed optimally, i.e. no potential is wasted between any two levels.

Recall from the proof of Lemma \ref{preservePot} that this is only possible if between any two consecutive switches of a node $v$, it is exactly a $\frac{\lambda+\varphi}{1-\varphi}$ fraction of $v$'s upper neighbors that switch. Moreover, these switching $\frac{\lambda+\varphi}{1-\varphi} \cdot \text{deg}(v)$ upper neighbors always have to be of the right color, i.e. they need to switch to the opposite of $v$'s current color in case of majority processes, and to the same color in case of minority processes. Since the upper neighbors of $v$ are in the same level, we also have to ensure that throughout the entire process, each upper neighbor switches the same number of times altogether.

These conditions impose heavy restrictions on the possible ways to connect two subsequent levels. If the conditions hold for a node $v$ (i.e. the sequence of switches of $v$'s upper neighbors can be split into $\frac{\lambda+\varphi}{1-\varphi} \cdot \text{deg}(v)$-size consecutive appropriate-colored subsets, in an altogether balanced way), then we say that $v$'s upper neighbors follow a valid \textit{control sequence}.

On the other hand, in order to argue about levels in general, we want each level to behave in a similar way. The easiest way to achieve this is to have a one-to-one correspondence between the nodes of different levels, and ensure that each level repeats the same sequence of steps periodically, but in a different pace. That is, we want to connect the levels in such a way that when a level exhibits a specific pattern of switches, then this allows the nodes of the next level to replicate the exact same pattern of switches, but more times.

Thus the key task in our lower bound constructions is to develop a so-called \emph{control gadget}, which is essentially a bipartite graph that fulfills these two requirements: it admits a scheduling of switches such that (i) the upper neighborhood of each lower node follows a valid control sequence, and (ii) while the upper level executes a sequence $s$ times, the lower level executes the same sequence $\frac{1-\varphi}{\lambda+\varphi} \cdot s$ times. Given such a control gadget, we can connect the subsequent level pairs of our construction using this gadgets. This allows us to indeed increase the number of switches by a $\frac{1-\varphi}{\lambda+\varphi}$ factor in each new level, resulting in a total of $\widetilde{\Theta}(n^{1+f(\lambda)})$ switches as described above.

However, developing a control gadget is a difficult combinatorial task in general: it depends on many factors including divisibility questions, and whether our parameters can be expressed as a fraction of small integers. A detailed discussion of control gadget design and the $\lambda$ values covered by Theorem \ref{the:lower} is available in Appendix \ref{App:B}. In particular, we present a method which allows us to develop a control gadget for every small $\lambda$ value below a threshold of approximately $0.476$ (more specifically, as long as $\frac{\lambda + \varphi}{1-\varphi} \leq \frac{3}{5}$). The same technique also provides a control gadget for some larger $\lambda$ values above the threshold, but only when the corresponding switch increase ratio $\frac{1-\varphi}{\lambda+\varphi}$ can be expressed as a fraction of relatively small integers. Furthermore, Appendix \ref{App:B} also describes a simpler solution technique to the control gadget problem; this leaves a slightly larger gap to the upper bound, but it works for any $\lambda$ without much difficulty.

\newpage

\bibliography{references}

\newpage

\begin{appendices}
\section{Discussion of upper bound proof} \label{App:A}

In this section, we discuss some parts of the upper bound proof in more detail.

\subsection{Majority and minority processes as CPSs}

When introducing the concept of CPS as the common abstraction of majority and minority processes, it is rather straightforward that conditions 2 and 3 are fulfilled, since each time when a node $v$ switches, it can only create 1 conflict on at most $\frac{1-\lambda}{2} \cdot \text{deg}(v)$ incident edges. Condition 1, however, requires some more discussion.

Between each two consecutive switches of $v$, we know that at least $\frac{1+\lambda}{2} \cdot \text{deg}(v)-\frac{1-\lambda}{2} \cdot \text{deg}(v)=\lambda \cdot \text{deg}(v)$ new conflicts must be wasted (i.e. removed) to raise the number of conflicts on incident edges above the switchability threshold of $\frac{1+\lambda}{2} \cdot \text{deg}(v)$ again. Furthermore, if between the two switches there are also conflicts that are removed from the incident edges by neighboring nodes (i.e., outputs), then each of these conflicts have to be replaced by a new one (an extra input) to have the required number of conflicts for switchability again.

More formally, let $in_i$ be the number of conflicts created on, and $out_i$ the number of conflicts removed from the edges of $v$ between the $(i-1)^{\text{th}}$ and $i^{\text{th}}$ switching of $v$, for $i \in \{1, ..., s(v)\}$. If $out_i$ further conflicts are removed from $v$'s edges before the $(i+1)^{\text{th}}$ switching of $v$, then $v$ needs to obtain $out_i$ further conflicts to reach the threshold of $\frac{1+\lambda}{2} \cdot \text{deg}(v)$ and be switchable for the $(i+1)^{\text{th}}$ time. This implies $in_i \geq \lambda \cdot \text{deg}(v) + out_i$; adding this up for all $i$ provides condition 1.

This explains why the relaxed version of condition 1 holds asymptotically. However, there are some edge cases that make the process slightly differ from this asymptotic behavior. Besides input conflicts (created by a neighbor of $v$), there may also be original conflicts on the edges incident to $v$, which were not created by a neighbor but were present from the beginning due to the initial coloring of the graph. These conflicts can be used by $v$ just like an input conflict when switching, and hence it is in fact the sum of original and input conflicts that has to be larger than the required number of conflicts for switching (i.e., the sum of outputs plus $\lambda \cdot \text{deg}(v) \cdot s(v)$). However, the number of original conflicts on incident edges is at most $\text{deg}(v)$, so adding an extra term of $\text{deg}(v)$ on the left side of condition 1 (i.e., requiring only that $c_{in}(v) \geq \lambda \cdot \text{deg}(v) \cdot s(v) + c_{out}(v) - \text{deg}(v)$) gives an inequality that holds for any node in a majority/minority process, even if a node $v$ uses up to $\text{deg}(v)$ original conflicts while switching. 

Also, the behavior of the process is slightly different before the first and after the last switch. On the one hand, in the first round, $v$ needs to use $\frac{1+\lambda}{2} \cdot \text{deg}(v)$ conflicts that are all inputs or original conflicts (whereas in later rounds, up to $\frac{1-\lambda}{2} \cdot \text{deg}(v)$ of the used conflicts might be ones that were created by $v$ in the previous round). Therefore, because of this first round, the total number of used conflicts is actually $\frac{1+\lambda}{2} \cdot \text{deg}(v)-\lambda \cdot \text{deg}(v)=\frac{1-\lambda}{2} \cdot \text{deg}(v)$ higher than in the asymptotic case. On the other hand, there is no need to compensate for output conflicts that are removed after the very last switching of $v$, since the number of conflicts in the final state of the graph is irrelevant; therefore, there may be up to $\frac{1-\lambda}{2} \cdot \text{deg}(v)$ output conflicts that do not have to be compensated.  Note, however, that these two edge cases do not require us to further modify condition 1, since the two new terms cancel each other on the right side.

\subsection{Relaxing the CPS definition}

While the extra $\text{deg}(v)$ term in condition 1 becomes asymptotically irrelevant if a node switches many times (i.e. $s(v)$ is large), the precise analysis still requires us to introduce the relaxed version of the CPS concept where condition 1 does not contain this extra term.

Consider a slightly smaller switching rule parameter $\lambda-\epsilon$, for any small $\epsilon>0$. Note that $c_{in}(v) \geq (\lambda-\epsilon) \cdot \text{deg}(v) \cdot s(v) + c_{out}(v)$ automatically implies $c_{in}(v) + \text{deg}(v) \geq \lambda \cdot \text{deg}(v) \cdot s(v) + c_{out}(v)$ for $s(v)$ large enough; that is, $\epsilon \cdot \text{deg}(v) \cdot s(v) \geq \text{deg}(v)$ holds whenever $s(v) \geq \frac{1}{\epsilon}$, so the additive term is not required. However, having $\lambda-\epsilon$ instead of $\lambda$ in the condition also results in the slightly less tight upper bound of $O(n^{1+f(\lambda-\epsilon)})$.

Therefore, we take the following approach. Assume we have a $\lambda_0$ for which we want to show the upper bound. We select a small $\epsilon>0$, and define $\lambda:=\lambda_0-\epsilon$. We define a constant switching threshold $s_0:=\frac{1}{\epsilon}$; nodes $v$ with $s(v) < s_0$ will be the base nodes. The base nodes in our graph then do not satisfy condition 1; however, since they only switch a few times, they have a limited influence on the process. By the choice of $s_0$, the remaining nodes satisfy condition 1 with $\lambda$, even without the extra term, so the relaxed version of condition 1 indeed holds with $s_0$ and $\lambda$.

We then follow the proof outlined before with Relaxed CPSs. This allows us to upper bound stabilization time by $O(n^{1+f(\lambda)})=O(n^{1+f(\lambda_0-\epsilon)})$. Since $f$ is continuous and the technique works for any $\epsilon>0$, this establishes an upper bound of $O(n^{1+f(\lambda_0)+\epsilon})$ for any $\epsilon>0$. Thus in terms of the parameter $\lambda_0$ of Rule II, our upper bound amounts to $O(n^{1+f(\lambda_0)+\epsilon})$ steps.

\subsection{Potential of dicuts}

Recall that Lemma \ref{preservePot} shows that the output potential of any node can be at most as much as its input potential. This allows us to upper bound the total potential in any dicut of the graph.

We use \textit{trivial dipartitioning} to refer to the dipartitioning ($V_1$, $V_2$) where $V_1$ only contains the source nodes of the DAG, and $V_2$ contains all other nodes.

\begin{lemma}
Every dipartitioning can be obtained from the trivial partitioning through a sequence of steps such that each intermediate step is also a dipartitioning.
\end{lemma}

\begin{proof}
The statement clearly holds for the trivial dipartitioning. For any other dipartitioning, we can prove the statement by induction on the number of nodes in $V_1$. Given any other dipartitioning ($V_1$, $V_2$), let us take a topological ordering of the DAG which begins with all the source nodes. Let us restrict this ordering to $V_1$, and let $v$ be the last node of the ordering which is in $V_1$. Since the ordering is topological, there are no edges from $v$ to $V_1 \setminus \{ v \}$. Therefore, ($V_1 \setminus \{ v \}$, $V_2 \cup \{ v \}$) is also a dipartitioning, so there exists a valid sequence to obtain it due to the induction hypothesis. Appending the dipartitioning ($V_1$, $V_2$) to the end of this sequence provides a sequence for ($V_1$, $V_2$).
\end{proof}

From this, the proof of Lemma \ref{cutbound} already follows. The dicut of the trivial dipartitioning has potential at most $O(n^2)$. Due to Lemma \ref{preservePot}, the potential of the dicut can only decrease throughout the sequence. This shows that the potential of dicut ($V_1$, $V_2$) is still at most as much potential of the trivial dipartitioning.

\subsection{Responsibility technique for the upper bound}

We now discuss the proof of Lemma \ref{lem:upper} in detail. Note that in the definition of a (relaxed) CPS, we defined the functions $s$ and $c$ as integer-valued, since this definition is intuitively closer to our original majority/minority processes. However, one can observe that all our statements in Section \ref{sec:upper} still hold if $s$ and $c$ are allowed to take any value among the nonnegative real numbers. Since allowing non-integer values allows for a simpler proof of Lemma \ref{lem:upper}, in the following, we consider this not-necessarily-integer version of CPSs in order to avoid some discretization challenges.

As an edge case, note that source nodes switch at most $O(1)$ time according to Lemma \ref{lem:constbase}, so altogether, they contribute at most $O(n)$ to the total number of switches. Therefore, we can ignore them in the analysis, and consider only the remaining nodes of the graph which satisfy the relaxed version of condition 1.

The main structure of the proof has already been outlined in Section \ref{sec:upperlast}; it only remains to describe the responsibility technique devised for the first part of the proof.

Let us take a topological ordering of the nodes in $A$, and let us iterate through the nodes of $A$ in this order. For each next node $v_0$ in this ordering, we define the responsibility of $v_0$, denoted $R(v_0)$. As outlined, we introduce a function $\Delta c(e)$ on the edges and $\Delta s(v)$ on the vertices for each such $v_0$, and after having processed $v_0$, we subtract these functions from $c(e)$ and $s(v)$, respectively.

That is, let $c': \widehat{E} \rightarrow \mathbb{R}$ and $s': \widehat{V} \rightarrow \mathbb{R}$, initially set to $c'(e):=c(e)$ and $s'(v):=s(v)$ for every vertex $v \in V$ and every directed edge $e$ of the DAG. Every time when we process the next node $v_0$, we define a new $\Delta c(e)$ and $\Delta s(v)$ based on the effects of $v_0$, and reduce $c'(e)$ by $\Delta c(e)$ on every $e \in \widehat{E}$, and reduce $s'(v)$ by $\Delta s(v)$ on every $v \in V$. Due to the definition of $\Delta c(e)$ and $\Delta s(v)$, the resulting $c'(e)$ and $s'(v)$ will still be a valid CPS after each step of the process. After processing all $v_0 \in A$, we obtain a final $c'(e)$ and $s'(v)$ for the second part of the proof outlined in Lemma \ref{lem:upper}.

\subsubsection{Definition of $\Delta c$ and $\Delta s$}

Let us now define the functions $\Delta c$ and $\Delta s$. Let $v_0$ be the next node of the topological ordering. In order to process the switches `caused by' $v_0$, we take a topological ordering of the nodes reachable from $v_0$ on the current edges of the DAG (that is, the real edges with regard to the current $c'(e)$). The first node of the ordering is clearly $v_0$ itself; for each output edge $(v_0, u) \in \widehat{E}$ of $v_0$, let $\Delta c(v_0, u)=c'(v_0, u)$. That is, after the current $\Delta c(e)$ will be subtracted from $c'(e)$, all output edges $(v_0, u)$ will have $c(v_0, u)=0$, and thus cease to be real edges, turning $v_0$ into a new sink node of the DAG.

In general, let $v$ be the next node in the topological ordering of the nodes reachable from $v_0$ (i.e., the inner loop of the algorithm). Since the ordering is topological, all input edges $(u,v)$ of $v$ already have a value $\Delta c(u,v)$ assigned to them (if an input node $u$ is not reachable from $v_0$, we consider $\Delta c(u,v)$ to have the default value of 0). Let $\Delta_{in}:=\sum_{(u,v) \in \widehat{E}} \Delta c(u,v)$.

First of all, we generally define
\begin{equation} \label{eq:deltadef1}
 \Delta s(v):=\frac{\Delta_{in}}{\frac{1+\lambda}{2} \cdot \text{deg}(v)}.
\end{equation}
Furthermore, we define $\Delta c(v,w)$ on the output edges $(v,w)$ of $v$ as follows. Similarly to the definition of $\Delta_{in}$, let $\Delta_{out}:=\sum_{(v,w) \in \widehat{E}} \Delta c(v,w)$. Our assignment will ensure two things. On the one hand, we assign $\Delta c(v,w)$ values such that $\Delta_{out} = \Delta s(v) \cdot \frac{1-\lambda}{2} \cdot \text{deg}(v)$; or, put otherwise through the definition of $\Delta s (v)$, $\Delta_{out}=\frac{1-\lambda}{1+\lambda} \cdot \Delta_{in}$. On the other hand, we always reduce the value $c'(v,w)$ on the output edge with the largest $c'(v,w)$ value, until a total reduction of $\frac{1-\lambda}{1+\lambda} \cdot \Delta_{in}$ is obtained.

Moreover, we have to apply a slightly different method when $c'_{out}(v) < \frac{1-\lambda}{1+\lambda} \cdot \Delta_{in}$, i.e. it is not large enough to be decreased by the required amount. In this case, we choose $\Delta_{out}$ as large as possible (that is, equal to $c'_{out}(v)$), and define $\widetilde{\Delta}_{in} = \Delta_{in} - \frac{\lambda+1}{\lambda-1} \cdot c'_{out}(v)$, i.e. the portion of the input which we cannot compensate from the remaining outputs. Since this part of the input conflicts is not used to create output conflicts, this can result in a higher number of switches at $v$. Hence, we reduce $s'(v)$ by a larger amount altogether. Specifically, we define
\begin{equation} \label{eq:deltadef2}
\Delta s(v):= \frac{\left( \Delta_{in} - \widetilde{\Delta}_{in} \right)}{\frac{1+\lambda}{2} \cdot \text{deg}(v)} + \frac{\widetilde{\Delta}_{in}}{\lambda \cdot \text{deg}(v)}.
\end{equation}

Intuitively, the idea behind this technique is that even if inputs are used in the most optimal format, then $1$ unit of input can correspond to at most $\frac{1-\lambda}{1+\lambda}$ units of output at $v$. This is because condition 2 ensures $c_{out}(v) \leq \frac{1-\lambda}{2} \cdot \text{deg}(v) \cdot s(v)$, and in case of the maximum possible output, condition 1 gives
\[ c_{in}(v) \geq \lambda \cdot \text{deg}(v) \cdot s(v) + \frac{1-\lambda}{2} \cdot \text{deg}(v) \cdot s(v) = \frac{1+\lambda}{2} \cdot \text{deg}(v) \cdot s(v),\]
providing a natural upper bound of $\frac{\frac{1+\lambda}{2}}{\frac{1-\lambda}{2}} = \frac{1+\lambda}{1-\lambda}$ on the rate of inputs to outputs. Furthermore, in case of this input to output ratio, the total input of (at least) $\frac{1+\lambda}{2} \cdot \text{deg}(v) \cdot s(v)$ corresponds to $s(v)$ switches, and thus each unit of input induces at most $\frac{1}{\frac{1+\lambda}{2} \cdot \text{deg}(v)}$ switches in $v$. On the other hand, when there are no more outputs anymore, the number of inputs $c_{in}(v)$ can be as low as $\lambda \cdot \text{deg}(v) \cdot s(v)$, and hence each unit of input induces at most $\frac{1}{\lambda \cdot \text{deg}(v)}$ switches in $v$.

To sum it up formally, when processing the next node $v$, we do the following. If $c'_{out}(v) \geq \frac{1-\lambda}{1+\lambda} \cdot \Delta_{in}$, then we define $\Delta s(v)$ according to Equation \ref{eq:deltadef1}. We select a threshold value $c_{thres}$, and define $\Delta c(v,w)$ on the output edges such that $\Delta c(v,w)=0$ for output edges where $c'(v,w) \leq c_{thres}$, and $\Delta c(v,w)=c'(v,w)-c_{thres}$ for output edges where $c'(v,w) > c_{thres}$. Since we can decrease $c_{thres}$ continuously, there exists exactly one threshold value which ensures that $\Delta_{out}=\frac{1-\lambda}{1+\lambda} \cdot \Delta_{in}$. Hence, each output $c'(v,w)$ is truncated to this threshold value.

Otherwise, if $c'_{out}(v) < \frac{1-\lambda}{1+\lambda} \cdot \Delta_{in}$, then we assign $\Delta c(v,w):=c'(v,w)$ to each output edge $(v,w)$ of $v$, calculate $\widetilde{\Delta}_{in}$ as discussed above, and define $\Delta s(v)$ according to Equation \ref{eq:deltadef2}.

\subsubsection{CPS conditions after subtracting $\Delta c$ and $\Delta s$}

\begin{lemma}
The definitions of these modifications ensure that after reducing the number of switches and conflicts, the resulting process still remains a CPS in each step.
\end{lemma}

\begin{proof}
Recall that the conditions of a relaxed CPS require
\renewcommand{\labelenumi}{\arabic{enumi}.}
\vspace{3pt}
\begin{enumerate}
{\setlength\itemindent{3pt} 
\item $c'_{in}(v) \geq \lambda \cdot \text{deg}(v) \cdot s'(v) + c'_{out}(v)$,
\item $c'_{out}(v) \leq \frac{1-\lambda}{2} \cdot \text{deg}(v) \cdot s'(v)$, and
\item $c'(v,w) \leq s'(v)$ for each output edge $(v,w)$
}
\end{enumerate}
\vspace{3pt}
for node $v$. We show that these conditions still hold for the new functions $c'$ and $s'$, obtained after subtracting $\Delta c$ and $\Delta s$.

First consider the case when there are still output $c'(v,w)$ values to decrease. In condition 1, the number of inputs decreases by $\Delta_{in}$ on the left side when executing the step. The number of outputs decreases by $\frac{1-\lambda}{1+\lambda} \cdot \Delta_{in}$ on the right side, and the first term on the right is reduced by
\[ \lambda \cdot \text{deg}(v) \cdot \Delta s(v) = \lambda \cdot \text{deg}(v) \cdot \frac{\Delta_{in}}{\frac{1+\lambda}{2} \cdot \text{deg}(v)} =  \frac{2 \lambda}{1+\lambda} \cdot \Delta_{in}. \]
This adds up to a decrease of $\left( \frac{1-\lambda}{1+\lambda} + \frac{2 \lambda}{1+\lambda} \right) \cdot \Delta_{in} = \Delta_{in}$ on the right side, thus condition 1 remains true in this case.

In condition 2, the left side is decreased by $\Delta_{out} = \frac{1-\lambda}{1+\lambda} \cdot \Delta_{in}$, while the right side is also decreased by
\[ \frac{1-\lambda}{2} \cdot \text{deg}(v) \cdot \Delta s(v) = \frac{1-\lambda}{2} \cdot \text{deg}(v) \cdot \frac{\Delta_{in}}{\frac{1+\lambda}{2} \cdot \text{deg}(v)} = \frac{1-\lambda}{1+\lambda} \cdot \Delta_{in} \]
in each step.

To show that condition 3 remains true, we use the fact that $c'(v,w)$ is always decreased on the output edges with the highest $c'(v,w)$ values. Assume that $c'(v,w_0) > s'(v)$ on some output edge $(v,w_0)$, for the new functions $c'$ and $s'$ obtained after subtracting $\Delta c$ and $\Delta s$. Recall that with our truncation technique, if we have $c'(v,w_0)$ on any edge after the reduction, then $c_{thres} \geq c'(v,w_0)$. Together, this implies $c_{thres} > s'(v)$.

Let $s_{prev}'(v):= s'(v) + \Delta s(v)$, the value of $s'(v)$ before the decrease. Recall that by the definition of $\Delta s(v)$, we have $s_{prev}'(v) - s'(v) = \Delta_{out} \cdot \frac{2}{1-\lambda} \cdot \frac{1}{\text{deg}(v)}$, so for the difference between $s_{prev}'(v)$ and $c_{thres}$, we have $s_{prev}'(v) - c_{thres} < \Delta_{out} \cdot \frac{2}{1-\lambda} \cdot \frac{1}{\text{deg}(v)}$. Note that this difference is the maximum value of $\Delta c(v, w)$ on any output edge, since before the decrease, all $c'(v,w)$ values were at most $s_{prev}'(v)$, and none of them were reduced below $c_{thres}$. However, since we decrease the outputs by $\Delta_{out}$ in total, this means that we have to reduce (i.e., have a nonzero $\Delta c(v, w)$) on strictly more than
\[ \frac{\Delta_{out}}{\Delta_{out} \cdot \frac{2}{1-\lambda} \cdot \frac{1}{\text{deg}(v)}} = \frac{1-\lambda}{2} \cdot \text{deg}(v) \]
distinct output edges. Each of these output edges is reduced to $c_{thres}$, so the total sum of outputs after the decrease is at least
\[ c'_{out}(v) \geq \frac{1-\lambda}{2} \cdot \text{deg}(v) \cdot c_{thres} > \frac{1-\lambda}{2} \cdot \text{deg}(v) \cdot s'(v), \]
which contradicts the already established condition 2. Thus condition 3 must also hold.

Finally, consider the other case, when there are no more output values $c'(v,w)$ to decrease. The left side of condition 1 is still reduced by $\Delta_{in}$, and the right side consists of the first term only, which is reduced by 
\[ \lambda \cdot \text{deg}(v) \cdot \Delta s(v) = \lambda \cdot \text{deg}(v) \cdot \frac{\Delta_{in}}{\lambda \cdot \text{deg}(v)} = \Delta_{in}, \] so condition 1 remains true. In this case, conditions 2 and 3 hold trivially, since all output edges $(v,w)$ already have $c'(v,w)=0$. \qedhere
\end{proof}

\subsubsection{Responsibilities of nodes}

Consider any $v_a \in A$ throughout the process. The value $s'(v_a)$ is initially equal to $s(v_a)$, and then keeps being reduced until $v_a$ is the next node in the topological ordering (i.e., when $v_0 = v_a$). From this point, $s'(v_a)$ is not changed anymore; on the other hand, when analyzing the effects of $v_a$, $s'(v)$ values of other nodes are reduced, and we reassign these switches to be the responsibility of $v_a$. That is, whenever having processed a node $v_0$, we define $R(v_0) = s'(v_0) + \sum_{v \in A} \Delta s(v)$ for the $\Delta s$ function obtained in case of this specific $v_0$. Clearly, throughout the process, every decrease $\Delta s$ happens with regard to a specific $v_0$, so this is indeed a redistribution of the original $s(v)$ values, and hence $\sum_{v \in A}{s(v)} = \sum_{v \in A}{R(v)}$ holds.

\begin{lemma}
For any $v_0 \in A$ and for the final $s'(v_0)$ value, we have $R(v_0) = O(s'(v_0))$.
\end{lemma}

\begin{proof}
Consider the round when $v_0$ is the chosen node in the outer loop. As said above, $s'(v_0)$ is not modified anymore after this round, so it already has its final value; also the value of $R(v_0)$ is decided solely in this round.

Since $v_0 \in A$, we have $\text{deg}(v_0)<2a$. Hence, according to condition 2, $c'_{out}(v_0) = \Delta_{out}(v_0) < \frac{1-\lambda}{2} \cdot 2a \cdot s'(v_0)$ at the beginning of this round. Note that at each node $v$ reachable from $v_0$, we have $\Delta_{out}(v) \leq \frac{1-\lambda}{1+\lambda} \cdot \Delta_{in}(v)$, and hence the total of amount of changes $\Delta c$ decreases by a constant factor at each node $v$. Hence after processing all nodes up to a distance of at most $d$, the total amount of changes $\Delta c$ on the edges is at most
\[ \Delta_{out}(v_0) \cdot \left( 1 + \frac{1-\lambda}{1+\lambda} + \left( \frac{1-\lambda}{1+\lambda} \right)^2 + ... + \left( \frac{1-\lambda}{1+\lambda} \right)^d \right). \]
Since this is a geometric series with $\frac{1-\lambda}{1+\lambda} < 1$, the total amount of changes is at most
\[ \Delta_{out}(v_0) \cdot \sum_{i=0}^{\infty} \left( \frac{1-\lambda}{1+\lambda} \right)^i \leq \Delta_{out}(v_0) \cdot \frac{1}{1-\frac{1-\lambda}{1+\lambda}} = \Delta_{out}(v_0) \cdot \frac{1+\lambda}{2 \cdot \lambda} \]
regardless of $d$, thus even when all the nodes reachable from $v_0$ have been processed. Note that at each node $v$, each unit of decrease in $\Delta_{in}(v)$ corresponds to either $\frac{2}{1+\lambda} \cdot \frac{1}{\text{deg}(v)}$ or $\frac{1}{\lambda} \cdot \frac{1}{\text{deg}(v)}$ decrease in $\Delta s(v)$ (depending on whether $v$ still has real output edges to decrease). Even if we take the larger decrease rate of $\frac{1}{\lambda} \cdot \frac{1}{\text{deg}(v)}$, this means that the total amount of changes $\Delta c$ can only produce a limited amount of total decrease $\Delta s$; more specifically
\[ \sum_{v \in A} \Delta s(v) \leq \Delta_{out}(v_0) \cdot \frac{1+\lambda}{2 \cdot \lambda} \cdot \frac{1}{\lambda} \cdot \frac{1}{\text{deg}(v)} \leq O(1) \cdot \frac{ \Delta_{out}(v_0)}{a}, \]
using the fact that each $v \in A$ has degree at least $a$. Thus using the upper bound $\Delta_{out}(v_0) \leq \frac{1-\lambda}{2} \cdot 2a \cdot s'(v_0)$, we get
\[ R(v_0) = s'(v_0) + \sum_{v \in A} \Delta s(v) \leq s'(v_0) + \frac{ O(1) \cdot \frac{1-\lambda}{2} \cdot 2a \cdot s'(v_0)}{a} =  s'(v_0) \cdot \left( 1 + O(1) \right) = O(s'(v_0)). \]
\qedhere
\end{proof}

Hence $\sum_{v \in A} s(v) = \sum_{v \in A}{R(v)} = O(\sum_{v \in A} s'(v))$, so it suffices to upper bound the sum of the final $s'(v)$ values in order to prove Lemma \ref{lem:upper}, as done in the second part of the proof in Section \ref{sec:upper}.

\section{Discussion of lower bound proof} \label{App:B}

We now discuss the main challenges of designing a control gadget, and present some techniques that allow a control gadget design for a wide range of $\lambda \in (0,1)$. Let us introduce the notation $\mu:=\frac{\lambda+\varphi}{1-\varphi}$ for the input switching rate.

\subsection{Lower bound construction for $\lambda=\frac{1}{3}$}

We first demonstrate the construction showing the tight lower bound for a specific $\lambda$ value of $\frac{1}{3}$. This choice of $\lambda$ has a range of advantages: both $f(\frac{1}{3})=\frac{1}{3}$ and the optimal output ratio $\varphi^*(\lambda)=\frac{1}{9}$ are rational, the ratio of inputs to outputs $\frac{1-\varphi}{\varphi}=8$ is an integer, and the number of switches also increases by an integer factor $\frac{1}{\mu}=\frac{1-\varphi}{\lambda+\varphi}=2$. Thanks to these properties, $\lambda=\frac{1}{3}$ allows a fairly simple control gadget design.

\begin{lemma}
Consider majority/minority processes under Rule II with $\lambda=\frac{1}{3}$. There exists a graph construction and initial coloring that has stabilization time $\widetilde{\Omega}(n^{4/3})$.
\end{lemma}

As outlined in Section \ref{sec:lower}, our construction consists of $L=\log_{8}(n)$ levels, each of which contains $\Theta(\frac{n}{\log{n}})$ nodes. Each consecutive pair of levels forms a regular bipartite graph, with $\frac{1}{8}$ of the degree of the previous consecutive pair. Each node $v$ has updegree $\frac{8}{9} \text{deg}(v)$ and downdegree $\frac{1}{9} \text{deg}(v)$.

E.g. in a majority process, in the initial state, $\frac{2}{8}$ of inputs will have the opposite color as $v$, and all other neighbors will have the same color. Whenever $\mu=\frac{1}{2}$ of the inputs (i.e., $\frac{4}{9}$ of the degree) switch to the opposite color, then $\frac{6}{8}$ of inputs will have the opposite color; as this is $\frac{6}{9}=\frac{1+\lambda}{2}$ of all neighbors, $v$ can now switch. As a result, the lower neighbors of $v$ will have a different color than $v$ (i.e., a conflict is pushed down), and eventually these nodes will follow $v$ to the same new color. This results in a state again where $\frac{2}{8}$ of inputs have the opposite color as $v$, and the rest have the same.

Note that between every two switches of $v$, exactly half of its upper neighbors switch, so the number of switches for each node will always increase by a factor of 2 if we move a level down. This shows that each node in the bottom level switches $2^L= n^{\frac{1}{3}}$ times. Since there are $\widetilde{\Theta}(n)$ nodes on the bottom level, the already sum up to $\widetilde{\Omega}(n^{4/3})$ switches, establishing the lower bound.

Two consecutive levels of the construction are connected through control gadgets. A control gadget is a regular, bipartite gadget on $k+k$ nodes for some constant $k$, i.e. a way to connect two $k$-tuples of nodes on a consecutive pair of levels. The upper and lower $k$ nodes of the gadget are in a 1-to-1 correspondence with each other. The goal of the gadget is to ensure that given some sequence of switches in the $k$-tuple, if we execute the the switches $s$ times on the upper level, then this allows us to execute the same sequence of switches on lower $k$-tuple $2s$ times. This allows for a recursive repetition of the same process, executed twice as many times on each next level.

\begin{figure}
\centering
\captionsetup{justification=centering}
	\begin{tikzpicture}

	\draw[gray, thick] (0pt,0pt) -- (22.5pt,20pt);
	\draw[gray, thick] (15pt,0pt) -- (22.5pt,20pt);
	\draw[gray, thick] (30pt,0pt) -- (22.5pt,20pt);
	\draw[gray, thick] (45pt,0pt) -- (22.5pt,20pt);
	\draw[gray, ultra thick] (22.5pt,20pt) -- (22.5pt,45pt);
	\draw[gray, thick] (0pt,90pt) -- (0pt,80pt) -- (22.5pt,45pt);
	\draw[gray, thick] (60pt,90pt) -- (60pt,80pt) -- (22.5pt,45pt);
	\draw[gray, thick] (120pt,90pt) -- (120pt,80pt) -- (22.5pt,45pt);
	\draw[gray, thick] (180pt,90pt) -- (180pt,80pt) -- (22.5pt,45pt);
	
	\draw[gray, thick] (60pt,0pt) -- (82.5pt,20pt);
	\draw[gray, thick] (75pt,0pt) -- (82.5pt,20pt);
	\draw[gray, thick] (90pt,0pt) -- (82.5pt,20pt);
	\draw[gray, thick] (105pt,0pt) -- (82.5pt,20pt);
	\draw[gray, ultra thick] (82.5pt,20pt) -- (82.5pt,45pt);
	\draw[gray, thick] (15pt,90pt) -- (15pt,80pt) -- (82.5pt,45pt);
	\draw[gray, thick] (75pt,90pt) -- (75pt,80pt) -- (82.5pt,45pt);
	\draw[gray, thick] (135pt,90pt) -- (135pt,80pt) -- (82.5pt,45pt);
	\draw[gray, thick] (195pt,90pt) -- (195pt,80pt) -- (82.5pt,45pt);
	
	\draw[gray, thick] (120pt,0pt) -- (142.5pt,20pt);
	\draw[gray, thick] (135pt,0pt) -- (142.5pt,20pt);
	\draw[gray, thick] (150pt,0pt) -- (142.5pt,20pt);
	\draw[gray, thick] (165pt,0pt) -- (142.5pt,20pt);
	\draw[gray, ultra thick] (142.5pt,20pt) -- (142.5pt,45pt);
	\draw[gray, thick] (30pt,90pt) -- (30pt,80pt) -- (142.5pt,45pt);
	\draw[gray, thick] (90pt,90pt) -- (90pt,80pt) -- (142.5pt,45pt);
	\draw[gray, thick] (150pt,90pt) -- (150pt,80pt) -- (142.5pt,45pt);
	\draw[gray, thick] (210pt,90pt) -- (210pt,80pt) -- (142.5pt,45pt);
	
	\draw[gray, thick] (180pt,0pt) -- (202.5pt,20pt);
	\draw[gray, thick] (195pt,0pt) -- (202.5pt,20pt);
	\draw[gray, thick] (210pt,0pt) -- (202.5pt,20pt);
	\draw[gray, thick] (225pt,0pt) -- (202.5pt,20pt);
	\draw[gray, ultra thick] (202.5pt,20pt) -- (202.5pt,45pt);
	\draw[gray, thick] (45pt,90pt) -- (45pt,80pt) -- (202.5pt,45pt);
	\draw[gray, thick] (105pt,90pt) -- (105pt,80pt) -- (202.5pt,45pt);
	\draw[gray, thick] (165pt,90pt) -- (165pt,80pt) -- (202.5pt,45pt);
	\draw[gray, thick] (225pt,90pt) -- (225pt,80pt) -- (202.5pt,45pt);
	
	\draw[gray, fill=gray] (22.5pt,20pt) circle (0.25ex);
	\draw[gray, fill=gray] (22.5pt,45pt) circle (0.25ex);
	\draw[gray, fill=gray] (82.5pt,20pt) circle (0.25ex);
	\draw[gray, fill=gray] (82.5pt,45pt) circle (0.25ex);
	\draw[gray, fill=gray] (142.5pt,20pt) circle (0.25ex);
	\draw[gray, fill=gray] (142.5pt,45pt) circle (0.25ex);
	\draw[gray, fill=gray] (202.5pt,20pt) circle (0.25ex);
	\draw[gray, fill=gray] (202.5pt,45pt) circle (0.25ex);
	

	\draw[black, fill=white] (0pt,0pt) circle (0.8ex);
	\draw[black, fill=white] (15pt,0pt) circle (0.8ex);
	\draw[black, fill=white] (30pt,0pt) circle (0.8ex);
	\draw[black, fill=white] (45pt,0pt) circle (0.8ex);
	\draw[black, fill=white] (60pt,0pt) circle (0.8ex);
	\draw[black, fill=white] (75pt,0pt) circle (0.8ex);
	\draw[black, fill=white] (90pt,0pt) circle (0.8ex);
	\draw[black, fill=white] (105pt,0pt) circle (0.8ex);
	\draw[black, fill=white] (120pt,0pt) circle (0.8ex);
	\draw[black, fill=white] (135pt,0pt) circle (0.8ex);
	\draw[black, fill=white] (150pt,0pt) circle (0.8ex);
	\draw[black, fill=white] (165pt,0pt) circle (0.8ex);
	\draw[black, fill=white] (180pt,0pt) circle (0.8ex);
	\draw[black, fill=white] (195pt,0pt) circle (0.8ex);
	\draw[black, fill=white] (210pt,0pt) circle (0.8ex);
	\draw[black, fill=white] (225pt,0pt) circle (0.8ex);
	
	\node[anchor=north] at (0pt,-2pt) {\tiny $A1$};
	\node[anchor=north] at (15pt,-2pt) {\tiny $B1$};
	\node[anchor=north] at (30pt,-2pt) {\tiny $C1$};
	\node[anchor=north] at (45pt,-2pt) {\tiny $D1$};
	\node[anchor=north] at (60pt,-2pt) {\tiny $A2$};
	\node[anchor=north] at (75pt,-2pt) {\tiny $B2$};
	\node[anchor=north] at (90pt,-2pt) {\tiny $C2$};
	\node[anchor=north] at (105pt,-2pt) {\tiny $D2$};
	\node[anchor=north] at (120pt,-2pt) {\tiny $A3$};
	\node[anchor=north] at (135pt,-2pt) {\tiny $B3$};
	\node[anchor=north] at (150pt,-2pt) {\tiny $C3$};
	\node[anchor=north] at (165pt,-2pt) {\tiny $D3$};
	\node[anchor=north] at (180pt,-2pt) {\tiny $A4$};
	\node[anchor=north] at (195pt,-2pt) {\tiny $B4$};
	\node[anchor=north] at (210pt,-2pt) {\tiny $C4$};
	\node[anchor=north] at (225pt,-2pt) {\tiny $D4$};
	
	
	\draw[black, fill=white] (0pt,90pt) circle (0.8ex);
	\draw[black, fill=white] (15pt,90pt) circle (0.8ex);
	\draw[black, fill=white] (30pt,90pt) circle (0.8ex);
	\draw[black, fill=white] (45pt,90pt) circle (0.8ex);
	\draw[black, fill=white] (60pt,90pt) circle (0.8ex);
	\draw[black, fill=white] (75pt,90pt) circle (0.8ex);
	\draw[black, fill=white] (90pt,90pt) circle (0.8ex);
	\draw[black, fill=white] (105pt,90pt) circle (0.8ex);
	\draw[black, fill=white] (120pt,90pt) circle (0.8ex);
	\draw[black, fill=white] (135pt,90pt) circle (0.8ex);
	\draw[black, fill=white] (150pt,90pt) circle (0.8ex);
	\draw[black, fill=white] (165pt,90pt) circle (0.8ex);
	\draw[black, fill=white] (180pt,90pt) circle (0.8ex);
	\draw[black, fill=white] (195pt,90pt) circle (0.8ex);
	\draw[black, fill=white] (210pt,90pt) circle (0.8ex);
	\draw[black, fill=white] (225pt,90pt) circle (0.8ex);
	
	\node[anchor=south] at (0pt,92pt) {\tiny $A1$};
	\node[anchor=south] at (15pt,92pt) {\tiny $B1$};
	\node[anchor=south] at (30pt,92pt) {\tiny $C1$};
	\node[anchor=south] at (45pt,92pt) {\tiny $D1$};
	\node[anchor=south] at (60pt,92pt) {\tiny $A2$};
	\node[anchor=south] at (75pt,92pt) {\tiny $B2$};
	\node[anchor=south] at (90pt,92pt) {\tiny $C2$};
	\node[anchor=south] at (105pt,92pt) {\tiny $D2$};
	\node[anchor=south] at (120pt,92pt) {\tiny $A3$};
	\node[anchor=south] at (135pt,92pt) {\tiny $B3$};
	\node[anchor=south] at (150pt,92pt) {\tiny $C3$};
	\node[anchor=south] at (165pt,92pt) {\tiny $D3$};
	\node[anchor=south] at (180pt,92pt) {\tiny $A4$};
	\node[anchor=south] at (195pt,92pt) {\tiny $B4$};
	\node[anchor=south] at (210pt,92pt) {\tiny $C4$};
	\node[anchor=south] at (225pt,92pt) {\tiny $D4$};

\end{tikzpicture}
	\caption{Illustration of the connections within the control gadget of 16+16 nodes for $\lambda=\frac{1}{3}$, with simplified notation for complete bipartite subgraphs on 4+4 nodes.}
	\label{fig:conns}
\end{figure}
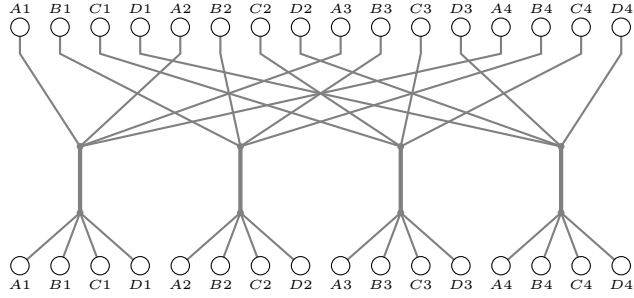

We present such a control gadget on $k=16$ nodes. For this, we take 4 groups $A, B, C, D$, each containing 4 nodes; thus, our nodes will be elements of $\{A, B, C, D\} \times \{ 1, 2, 3, 4\}$. Each lower level node labeled by number $x$ will be connected to the group corresponding to the $x^{\text{th}}$ letter of the alphabet. E.g. nodes $A2$, $B2$, $C2$ and $D2$ on the lower level form a complete bipartite subgraph with nodes $B1$, $B2$, $B3$ and $B4$ on the upper level; the connections are illustrated on Figure \ref{fig:conns}. Hence, each node has an induced degree 4 within the gadget.

Given these connections, Figure \ref{fig:whole_sequence} shows a self-replicating sequence of this control gadget. Considering the 4 upper neighbors of any specific node (without the group identifier), we can see that they follow the control sequence (12)(23)(34)(41). This ensures that every node occurs the same number of times in the sequence, and that between any two switches of a lower node, exactly 2 of its 4 upper neighbors are switched, so no inputs are wasted indeed. (Note that the simpler sequence (12)(34) would also satisfy these properties, but it would not allow us to assign colors to the nodes in a proper way.)

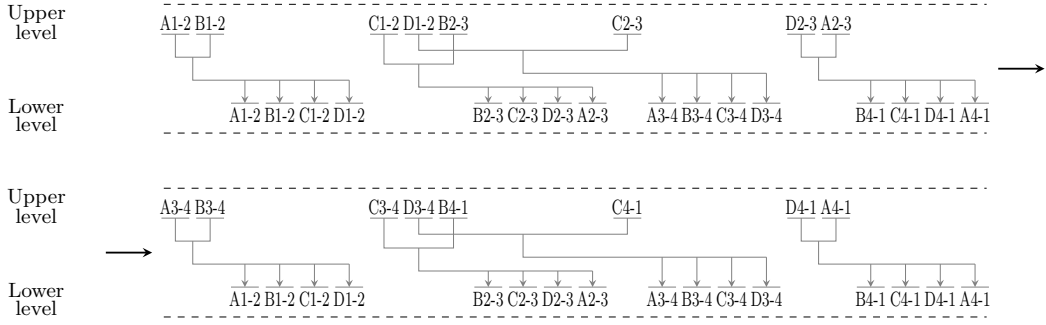
\begin{figure}
\centering
\captionsetup{justification=centering}
	\resizebox{0.99\textwidth}{!}{\begin{tikzpicture}

	\node[anchor=center] at (0pt,40pt) {\small \scalebox{.7}[1.0]{A1-2}};
	\node[anchor=center] at (15pt,40pt) {\small \scalebox{.7}[1.0]{B1-2}};
	\node[anchor=center] at (30pt,0pt) {\small \scalebox{.7}[1.0]{A1-2}};
	\node[anchor=center] at (45pt,0pt) {\small \scalebox{.7}[1.0]{B1-2}};
	\node[anchor=center] at (60pt,0pt) {\small \scalebox{.7}[1.0]{C1-2}};
	\node[anchor=center] at (75pt,0pt) {\small \scalebox{.7}[1.0]{D1-2}};
	\node[anchor=center] at (90pt,40pt) {\small \scalebox{.7}[1.0]{C1-2}};
	\node[anchor=center] at (105pt,40pt) {\small \scalebox{.7}[1.0]{D1-2}};
	\node[anchor=center] at (120pt,40pt) {\small \scalebox{.7}[1.0]{B2-3}};
	\node[anchor=center] at (135pt,0pt) {\small \scalebox{.7}[1.0]{B2-3}};
	\node[anchor=center] at (150pt,0pt) {\small \scalebox{.7}[1.0]{C2-3}};
	\node[anchor=center] at (165pt,0pt) {\small \scalebox{.7}[1.0]{D2-3}};
	\node[anchor=center] at (180pt,0pt) {\small \scalebox{.7}[1.0]{A2-3}};
	\node[anchor=center] at (195pt,40pt) {\small \scalebox{.7}[1.0]{C2-3}};
	\node[anchor=center] at (210pt,0pt) {\small \scalebox{.7}[1.0]{A3-4}};
	\node[anchor=center] at (225pt,0pt) {\small \scalebox{.7}[1.0]{B3-4}};
	\node[anchor=center] at (240pt,0pt) {\small \scalebox{.7}[1.0]{C3-4}};
	\node[anchor=center] at (255pt,0pt) {\small \scalebox{.7}[1.0]{D3-4}};
	\node[anchor=center] at (270pt,40pt) {\small \scalebox{.7}[1.0]{D2-3}};
	\node[anchor=center] at (285pt,40pt) {\small \scalebox{.7}[1.0]{A2-3}};
	\node[anchor=center] at (300pt,0pt) {\small \scalebox{.7}[1.0]{B4-1}};
	\node[anchor=center] at (315pt,0pt) {\small \scalebox{.7}[1.0]{C4-1}};
	\node[anchor=center] at (330pt,0pt) {\small \scalebox{.7}[1.0]{D4-1}};
	\node[anchor=center] at (345pt,0pt) {\small \scalebox{.7}[1.0]{A4-1}};
	
	\node[anchor=center] at (0pt,-40pt) {\small \scalebox{.7}[1.0]{A3-4}};
	\node[anchor=center] at (15pt,-40pt) {\small \scalebox{.7}[1.0]{B3-4}};
	\node[anchor=center] at (30pt,-80pt) {\small \scalebox{.7}[1.0]{A1-2}};
	\node[anchor=center] at (45pt,-80pt) {\small \scalebox{.7}[1.0]{B1-2}};
	\node[anchor=center] at (60pt,-80pt) {\small \scalebox{.7}[1.0]{C1-2}};
	\node[anchor=center] at (75pt,-80pt) {\small \scalebox{.7}[1.0]{D1-2}};
	\node[anchor=center] at (90pt,-40pt) {\small \scalebox{.7}[1.0]{C3-4}};
	\node[anchor=center] at (105pt,-40pt) {\small \scalebox{.7}[1.0]{D3-4}};
	\node[anchor=center] at (120pt,-40pt) {\small \scalebox{.7}[1.0]{B4-1}};
	\node[anchor=center] at (135pt,-80pt) {\small \scalebox{.7}[1.0]{B2-3}};
	\node[anchor=center] at (150pt,-80pt) {\small \scalebox{.7}[1.0]{C2-3}};
	\node[anchor=center] at (165pt,-80pt) {\small \scalebox{.7}[1.0]{D2-3}};
	\node[anchor=center] at (180pt,-80pt) {\small \scalebox{.7}[1.0]{A2-3}};
	\node[anchor=center] at (195pt,-40pt) {\small \scalebox{.7}[1.0]{C4-1}};
	\node[anchor=center] at (210pt,-80pt) {\small \scalebox{.7}[1.0]{A3-4}};
	\node[anchor=center] at (225pt,-80pt) {\small \scalebox{.7}[1.0]{B3-4}};
	\node[anchor=center] at (240pt,-80pt) {\small \scalebox{.7}[1.0]{C3-4}};
	\node[anchor=center] at (255pt,-80pt) {\small \scalebox{.7}[1.0]{D3-4}};
	\node[anchor=center] at (270pt,-40pt) {\small \scalebox{.7}[1.0]{D4-1}};
	\node[anchor=center] at (285pt,-40pt) {\small \scalebox{.7}[1.0]{A4-1}};
	\node[anchor=center] at (300pt,-80pt) {\small \scalebox{.7}[1.0]{B4-1}};
	\node[anchor=center] at (315pt,-80pt) {\small \scalebox{.7}[1.0]{C4-1}};
	\node[anchor=center] at (330pt,-80pt) {\small \scalebox{.7}[1.0]{D4-1}};
	\node[anchor=center] at (345pt,-80pt) {\small \scalebox{.7}[1.0]{A4-1}};

	\draw[gray] (-6pt,35pt) -- (6pt,35pt);
	\draw[gray] (9pt,35pt) -- (21pt,35pt);
	\draw[gray] (0pt,35pt) -- (0pt,25pt) -- (15pt,25pt) -- (15pt,35pt);
	\draw[gray] (7.5pt,25pt) -- (7.5pt,15pt) -- (75pt,15pt);
	\draw[gray, arrows=-stealth] (30pt,15pt) -- (30pt,5pt);
	\draw[gray, arrows=-stealth] (45pt,15pt) -- (45pt,5pt);
	\draw[gray, arrows=-stealth] (60pt,15pt) -- (60pt,5pt);
	\draw[gray, arrows=-stealth] (75pt,15pt) -- (75pt,5pt);
	\draw[gray] (24pt,5pt) -- (36pt,5pt);
	\draw[gray] (39pt,5pt) -- (51pt,5pt);
	\draw[gray] (54pt,5pt) -- (66pt,5pt);
	\draw[gray] (69pt,5pt) -- (81pt,5pt);
	
	\draw[gray] (84pt,35pt) -- (96pt,35pt);
	\draw[gray] (114pt,35pt) -- (126pt,35pt);
	\draw[gray] (90pt,35pt) -- (90pt,22pt) -- (120pt,22pt) -- (120pt,35pt);
	\draw[gray] (105pt,22pt) -- (105pt,12pt) -- (180pt,12pt);
	\draw[gray, arrows=-stealth] (135pt,12pt) -- (135pt,5pt);
	\draw[gray, arrows=-stealth] (150pt,12pt) -- (150pt,5pt);
	\draw[gray, arrows=-stealth] (165pt,12pt) -- (165pt,5pt);
	\draw[gray, arrows=-stealth] (180pt,12pt) -- (180pt,5pt);
	\draw[gray] (129pt,5pt) -- (141pt,5pt);
	\draw[gray] (144pt,5pt) -- (156pt,5pt);
	\draw[gray] (159pt,5pt) -- (171pt,5pt);
	\draw[gray] (174pt,5pt) -- (186pt,5pt);
	
	\draw[gray] (99pt,35pt) -- (111pt,35pt);
	\draw[gray] (189pt,35pt) -- (201pt,35pt);
	\draw[gray] (105pt,35pt) -- (105pt,28pt) -- (195pt,28pt) -- (195pt,35pt);
	\draw[gray] (150pt,28pt) -- (150pt,18pt) -- (255pt,18pt);
	\draw[gray, arrows=-stealth] (210pt,18pt) -- (210pt,5pt);
	\draw[gray, arrows=-stealth] (225pt,18pt) -- (225pt,5pt);
	\draw[gray, arrows=-stealth] (240pt,18pt) -- (240pt,5pt);
	\draw[gray, arrows=-stealth] (255pt,18pt) -- (255pt,5pt);
	\draw[gray] (204pt,5pt) -- (216pt,5pt);
	\draw[gray] (219pt,5pt) -- (231pt,5pt);
	\draw[gray] (234pt,5pt) -- (246pt,5pt);
	\draw[gray] (249pt,5pt) -- (261pt,5pt);
	
	\draw[gray] (264pt,35pt) -- (276pt,35pt);
	\draw[gray] (279pt,35pt) -- (291pt,35pt);
	\draw[gray] (270pt,35pt) -- (270pt,25pt) -- (285pt,25pt) -- (285pt,35pt);
	\draw[gray] (277.5pt,25pt) -- (277.5pt,15pt) -- (345pt,15pt);
	\draw[gray, arrows=-stealth] (300pt,15pt) -- (300pt,5pt);
	\draw[gray, arrows=-stealth] (315pt,15pt) -- (315pt,5pt);
	\draw[gray, arrows=-stealth] (330pt,15pt) -- (330pt,5pt);
	\draw[gray, arrows=-stealth] (345pt,15pt) -- (345pt,5pt);
	\draw[gray] (294pt,5pt) -- (306pt,5pt);
	\draw[gray] (309pt,5pt) -- (321pt,5pt);
	\draw[gray] (324pt,5pt) -- (336pt,5pt);
	\draw[gray] (339pt,5pt) -- (351pt,5pt);
	
	\draw[gray] (-6pt,-45pt) -- (6pt,-45pt);
	\draw[gray] (9pt,-45pt) -- (21pt,-45pt);
	\draw[gray] (0pt,-45pt) -- (0pt,-55pt) -- (15pt,-55pt) -- (15pt,-45pt);
	\draw[gray] (7.5pt,-55pt) -- (7.5pt,-65pt) -- (75pt,-65pt);
	\draw[gray, arrows=-stealth] (30pt,-65pt) -- (30pt,-75pt);
	\draw[gray, arrows=-stealth] (45pt,-65pt) -- (45pt,-75pt);
	\draw[gray, arrows=-stealth] (60pt,-65pt) -- (60pt,-75pt);
	\draw[gray, arrows=-stealth] (75pt,-65pt) -- (75pt,-75pt);
	\draw[gray] (24pt,-75pt) -- (36pt,-75pt);
	\draw[gray] (39pt,-75pt) -- (51pt,-75pt);
	\draw[gray] (54pt,-75pt) -- (66pt,-75pt);
	\draw[gray] (69pt,-75pt) -- (81pt,-75pt);
	
	\draw[gray] (84pt,-45pt) -- (96pt,-45pt);
	\draw[gray] (114pt,-45pt) -- (126pt,-45pt);
	\draw[gray] (90pt,-45pt) -- (90pt,-58pt) -- (120pt,-58pt) -- (120pt,-45pt);
	\draw[gray] (105pt,-58pt) -- (105pt,-68pt) -- (180pt,-68pt);
	\draw[gray, arrows=-stealth] (135pt,-68pt) -- (135pt,-75pt);
	\draw[gray, arrows=-stealth] (150pt,-68pt) -- (150pt,-75pt);
	\draw[gray, arrows=-stealth] (165pt,-68pt) -- (165pt,-75pt);
	\draw[gray, arrows=-stealth] (180pt,-68pt) -- (180pt,-75pt);
	\draw[gray] (129pt,-75pt) -- (141pt,-75pt);
	\draw[gray] (144pt,-75pt) -- (156pt,-75pt);
	\draw[gray] (159pt,-75pt) -- (171pt,-75pt);
	\draw[gray] (174pt,-75pt) -- (186pt,-75pt);
	
	\draw[gray] (99pt,-45pt) -- (111pt,-45pt);
	\draw[gray] (189pt,-45pt) -- (201pt,-45pt);
	\draw[gray] (105pt,-45pt) -- (105pt,-52pt) -- (195pt,-52pt) -- (195pt,-45pt);
	\draw[gray] (150pt,-52pt) -- (150pt,-62pt) -- (255pt,-62pt);
	\draw[gray, arrows=-stealth] (210pt,-62pt) -- (210pt,-75pt);
	\draw[gray, arrows=-stealth] (225pt,-62pt) -- (225pt,-75pt);
	\draw[gray, arrows=-stealth] (240pt,-62pt) -- (240pt,-75pt);
	\draw[gray, arrows=-stealth] (255pt,-62pt) -- (255pt,-75pt);
	\draw[gray] (204pt,-75pt) -- (216pt,-75pt);
	\draw[gray] (219pt,-75pt) -- (231pt,-75pt);
	\draw[gray] (234pt,-75pt) -- (246pt,-75pt);
	\draw[gray] (249pt,-75pt) -- (261pt,-75pt);
	
	\draw[gray] (264pt,-45pt) -- (276pt,-45pt);
	\draw[gray] (279pt,-45pt) -- (291pt,-45pt);
	\draw[gray] (270pt,-45pt) -- (270pt,-55pt) -- (285pt,-55pt) -- (285pt,-45pt);
	\draw[gray] (277.5pt,-55pt) -- (277.5pt,-65pt) -- (345pt,-65pt);
	\draw[gray, arrows=-stealth] (300pt,-65pt) -- (300pt,-75pt);
	\draw[gray, arrows=-stealth] (315pt,-65pt) -- (315pt,-75pt);
	\draw[gray, arrows=-stealth] (330pt,-65pt) -- (330pt,-75pt);
	\draw[gray, arrows=-stealth] (345pt,-65pt) -- (345pt,-75pt);
	\draw[gray] (294pt,-75pt) -- (306pt,-75pt);
	\draw[gray] (309pt,-75pt) -- (321pt,-75pt);
	\draw[gray] (324pt,-75pt) -- (336pt,-75pt);
	\draw[gray] (339pt,-75pt) -- (351pt,-75pt);

	\node[anchor=center] at (-60pt,44pt) {\small Upper};
	\node[anchor=center] at (-60pt,36pt) {\small level};
	\node[anchor=center] at (-60pt,4pt) {\small Lower};
	\node[anchor=center] at (-60pt,-4pt) {\small level};
	\node[anchor=center] at (-60pt,-36pt) {\small Upper};
	\node[anchor=center] at (-60pt,-44pt) {\small level};
	\node[anchor=center] at (-60pt,-76pt) {\small Lower};
	\node[anchor=center] at (-60pt,-84pt) {\small level};
	
	\draw[thick, arrows=-stealth] (355pt,20pt) -- (375pt,20pt);
	\draw[thick, arrows=-stealth] (-30pt,-60pt) -- (-10pt,-60pt);
	
	\draw[dashed] (-5pt,48pt) -- (350pt,48pt);
	\draw[dashed] (-5pt,-8pt) -- (350pt,-8pt);
	\draw[dashed] (-5pt,-32pt) -- (350pt,-32pt);
	\draw[dashed] (-5pt,-88pt) -- (350pt,-88pt);


\end{tikzpicture}}
	\caption{Self-replicating sequence of switches on 16 nodes: while the upper level executes the sequence once, the lower level executes the same sequence twice. Arrows show that the lower nodes become switchable due to the switching of the specific upper nodes.}
	\label{fig:whole_sequence}
\end{figure}

Having designed this control gadget of constant size, each level will consist of $\Theta(\frac{n}{\log{n}})$ distinct copies of this 16-node group $\{A, B, C, D\} \times \{ 1, 2, 3, 4\}$.  We then start with constant-degree nodes on the lowermost level, and increase this degree by a factor of $\frac{1-\varphi}{\varphi}=8$ on every new level from bottom to top. To achieve this degree, we connect the lower level of a control gadget to the upper level of not only one, but \textit{multiple} control gadgets; e.g. the nodes $A2$, $B2$, $C2$ and $D2$ are connected to the $B$-labeled nodes of not only one, but multiple 16-node groups on the level above. This allows us to indeed increase the degree by a factor of 8 at each new level. For example, if the node $A2$ in a group is connected to the nodes $A1$, $B1$, $C1$ and $D1$ in $x$ distinct 16-node groups on the level below (thus having a downdegree of $4x$), it will be connected to the nodes $B1$, $B2$, $B3$ and $B4$ in $8x$ distinct 16-node groups on the level above (resulting in an updegree of $32x$).

Since all 16-node groups on the same level can execute the same steps in a parallel manner, this allows us to produce the very same behavior as in the control gadget, but for high-degree nodes.  With this technique, each consecutive pair of levels will form a regular (i.e., same-degree) bipartite graph, comprised of numerous copies of the control gadget as a subgraph.

Given the construction for propagating conflicts appropriately, we can easily assign colors to the nodes to obtain a majority or minority process. Observe that a constructions for majority and minority processes follow straightforwardly from each other: since our graph is bipartite, we can simply reverse the color of every node on every second level, directly obtaining a minority example from a majority example, or vice versa.

\subsection{Generalization for other $\lambda$ values}

The main idea for generalizing the construction, as already outlined in Section \ref{sec:lower}, is the following. Given a control gadget of constant size, we can place $\Theta(\frac{n}{\log{n}})$ such gadgets on each level, having $L=\frac{1}{\log(\frac{1-\varphi}{\varphi})} \log(n)$ levels altogether. We then begin with a constant degree for each node on the lowermost level, and increase the degree by a factor of $\frac{1-\varphi}{\varphi}$ on each new level. In order to do this, we again connect the lower level of control gadgets to the upper level of not only one, but multiple distinct control gadgets, as in the case of the $\lambda=\frac{1}{3}$ example. Thus consecutive pairs of levels form a regular bipartite graph, with the degree rising exponentially as we move upward in the construction.

The main challenge in the general construction is to design a control gadget of constant size, i.e. to devise a way where the next level of nodes follows the exact some switching order, but with a schedule where the nodes switch an $\frac{1}{\mu}$ factor more frequently. However, when the input switching rate $\mu$ is not a rational number, then switching a $\mu$ portion of the upper neighborhood is of course not possible. Hence in this case, we can only approximate the rate by a rational number $\frac{p}{q} \approx \mu$, with $p, q, \in \mathbb{Z}$. With the appropriate choice of $p$ and $q$, we can get arbitrarily close to the desired rate $\mu$. We then develop the same construction and control gadget for the input switching rate $\frac{p}{q}$, which will yield almost the same amount of total switches: since $f(\lambda)$ is continuous, a close enough $\frac{p}{q}$ approximation gives a construction with $\Theta(n^{1+f(\lambda)-\epsilon})$ switches for any $\epsilon>0$.

For convenience, we will always assume that $p+q$ is an even value; in case it is not, we can easily achieve this by doubling the value of both $p$ and $q$, using the approximation $\frac{2p}{2q} \approx \mu$ instead of $\frac{p}{q}$. Note that in the the previous subsection where $\lambda=\frac{1}{3}$ implied $\mu=\frac{1}{2}$, we have already done this essentially: while we could have switched 1 out of 2 upper neighbors in each step, we have in fact switched 2 out of 4 every time. This assumption is required because we want nodes to be in conflict with $\frac{p+q}{2}$ out of their $q$ upper neighbors when switching, since this is the amount of upper neighbors that correspond to the switching threshold, namely
\[ \frac{\frac{p+q}{2}}{q} \cdot \text{deg}_{\text{upper}}(v) = \frac{1}{2} \cdot \frac{p+q}{q} \cdot (1-\varphi) \cdot \text{deg}(v) \approx \left( \frac{\lambda+\varphi}{1-\varphi} + 1 \right) \cdot \frac{1}{2} (1-\varphi) \cdot \text{deg}(v) = \frac{\lambda+1}{2} \cdot \text{deg}(v). \]
Hence, $\frac{p+q}{2}$ has to be an integer.

In the following, in order to develop the required control gadget, we first generalize the notion of control sequence for any $(p,q)$ pair; this is essentially a balanced schedule of switching in the upper neighborhood which ensures wasteless conflict propagation, i.e. that the lower neighbor always switches when it is exactly on the threshold of switchability. We then discuss the main challenge in generalizing the control gadget used for $\lambda=\frac{1}{3}$ to other $\lambda$ values.

Furthermore, the construction also raises some minor technical questions relating to divisibility; we discuss these at the end of the section.

\subsection{Control sequences for general $p$ and $q$}

Similarly to the $\mu=\frac{1}{2}$ case, given $p$ and $q$, we can develop a control sequence of numbers $(1, ..., q)$, and switch the upper neighborhood of any node in our construction following this sequence. Let $b=\frac{p-q}{2}$. The first \textit{bracket} of the control sequence contains numbers $(1, ..., p)$, and for every next bracket, we shift the both the beginning and the end of the interval by $b$; in general, the $i^{\text{th}}$ bracket consist of the numbers $((i-1) \cdot b + 1)$, ..., $((i-1) \cdot b + p)$, all taken modulo $q$ to fall into the interval $[1, ..., q]$.

Initially, all nodes labeled $1, ..., p$ and $p+b+1, ..., q$ are black, and all nodes labeled $p+1, ..., p+b$ are white. Then this sequence of steps ensures that in every odd step, all the nodes in the next bracket of the control sequence are currently black, and in every even step, all the nodes in the next bracket are currently white. This means that after every odd (or even) step, $\frac{p+b}{q}$ of the upper neighborhood is white (or black, respectively). As
\[ \frac{p+b}{q}=\mu+\frac{1-\mu}{2}=\frac{1+\mu}{2}=\frac{1+\lambda}{2(1-\varphi)} ,\]
and all output connections have a non-conflicting color before switching, this means that $(1-\varphi) \cdot \frac{1+\lambda}{2(1-\varphi)} = \frac{1+\lambda}{2}$ of the entire neighborhood is in conflict with the node, so it is indeed precisely on the threshold for switchability.

For example, the control sequence for ($p$,$q$)=(5,9) is
\[ (12345)(34567)(56789)(78912)(91234)(23456)(45678)(67891)(89123) ,\]
with nodes labeled 1-5 and 8-9 initially black and nodes labeled 6-7 initially white. Then in every odd (even)  bracket, the nodes that switch are always colored black (white) currently. To some extent, the same control sequence idea has already been applied in \cite{minorityW}.

Since $b$ and $q$ are relatively prime (as the greatest common divisor of $p$ and $q$ is either 1 or 2), the sequence consists of $q$ distinct brackets before periodically repeating itself. Note that among the nodes of a specific color, the next bracket always includes those that have occurred the least amount of times so far (have the smallest \textit{occurrence number}). This ensures that at any point in the sequence, the difference in the number of occurrences between any two nodes is at most 2. Whenever a specific node is absent from the sequence, it is always absent for exactly 2 consecutive brackets. Each node $1,...,q$ appears the same number of times ($p$ times) before the sequence start repeating itself; hence, if the upper neighborhood of a node $v$ follows this sequence, then $v$ indeed switches $\frac{q}{p}=\frac{1}{\mu}$ times more than its upper neighbors, and does not waste any input conflicts.

Observe, however, that any node $v$ connected to such an upper neighborhood has to be of the same color to be switchable in all steps. I.e. in case of a majority process, $v$ becomes white (black) after every odd step (even step, respectively), while in a minority process, $v$ becomes black (white) after every odd step (even step, respectively). Since we also need nodes of both color on the next level, in practice, we have to take two copies of our control gadgets; this produces twice as many nodes on each level, distributed equally among the two colors, which all switch at the same time if we proceed through the steps of the two control gadgets in a parallel manner. This technique of duplicating the controlling gadget has already been used and discussed in \cite{minorityW}. The duplication is a technical step that increases the size of each level by a factor of 2 only; hence in the following, we do not consider the color of nodes, and instead focus on the main challenge, which is the design of the control gadget that is to be duplicated.

\subsection{From control sequence to control gadget}

In our example for ($p$,$q$)=(2,4), we created 4 groups ($A-D$) of 4 nodes each ($1-4$). At specific points in time, in every group of the upper level, two nodes become switchable (at the same time in each upper group). We then process these upper groups in a permutation of our choice: in each step, we select one of these groups (a `letter'), and switch 2 nodes in this group, according to the next bracket of the control sequence. We will refer to such a step as \textit{switching the group}; note that this does not mean switching all nodes in the group, but executing a step of the control sequence, i.e., switching $\mu$ portion of the group so that all lower neighbors of the group become switchable. Once all four groups have been switched, all 16 nodes on the lower level become switchable, so we can start (or continue) executing the same process on the level-pair below.
 
Note that on the upper level, each next step in a specific group always picks a predetermined pair of nodes in the group (based on the control sequence), so in the upper level, it is enough to consider the order in which we select the groups: regardless of the actual nodes switched, the step always has the same effect, namely, it makes all nodes connected to this group switchable. In contrast to this, on the lower level, all nodes labeled with the same number become switchable at the same time, as they have the same upper neighbors (a specific group); thus when discussing the switchability of lower-level nodes, we can simply handle the nodes labeled with the same number together. Thus we can illustrate the process in a simplified way in the following diagram (note that numbers within the brackets of the control sequence are only reordered for better visibility).
\begin{figure}[H]
\centering
	\vspace{-2pt}
	\begin{tikzpicture}

	\draw[gray, very thick] (-5pt,26pt) -- (-4pt,4pt);
	\draw[gray, very thick] (8pt,26pt) -- (4pt,4pt);
	\draw[gray, very thick] (22pt,26pt) -- (28pt,4pt);
	\draw[gray, very thick] (35pt,26pt) -- (62pt,4pt);
	\draw[gray, very thick] (60.5pt,26pt) -- (34pt,4pt);
	\draw[gray, very thick] (72.5pt,26pt) -- (68pt,4pt);
	\draw[gray, very thick] (87.5pt,26pt) -- (92pt,4pt);
	\draw[gray, very thick] (100.5pt,26pt) -- (100pt,4pt);
	
	\draw[white, fill=white] (-10pt,24.5pt) -- (-10pt,36pt) -- (106pt,36pt) -- (106pt,24.5pt) -- cycle;
	\node[anchor=center] at (16pt,30pt) {$ A \;\, B \;\; C \;\, D $};
	\node[anchor=center] at (80pt,30pt) {$ B \;\, C \;\; A \;\, D $};
	\node[anchor=center] at (48pt,27pt) {\Large .};
	
	\draw[white, fill=white] (-10pt,-6pt) -- (-10pt,6pt) -- (106pt,6pt) -- (106pt,-6pt) -- cycle;
	\node[anchor=center] at (0pt,0pt) {$( \, 1 \; 2 \,)$};
	\node[anchor=center] at (32pt,0pt) {$( \, 3 \; 2 \,)$};
	\node[anchor=center] at (64pt,0pt) {$( \, 4 \; 3 \,)$};
	\node[anchor=center] at (96pt,0pt) {$( \, 1 \; 4 \,)$};

\end{tikzpicture}
	\vspace{-7pt}
\end{figure}
Note that when processing the second bracket, we need to switch group $B$ for the second time. Before that, we first execute the first switching of group $D$, too, and then by reaching up to the level above the upper level, we make all four groups switchable for the second time (denoted by a dot in the figure), and then switch $B$ for the second time. Note that this first switch of group $D$ already makes the nodes labeled 4 switchable when processing the second bracket. This is not a problem; since number 4 is not in the second bracket, we simply wait with the switching of these nodes until we start processing the third bracket.

Also note that we always ensure that the nodes of a specific bracket (e.g., nodes labeled 3 and 4 in the previous example) are all switched at the same time. This is needed to carry our initial the assumption over to the level-pair below, namely that the upper groups all become switchable together at specific points, and we can switch them in any order of our choice.

It is a natural idea to generalize this method for any ($p$,$q$) pair, by creating $q$ different groups of $q$ nodes each, and cross-connecting these $q^2$ nodes in a similar fashion. However, it is not straightforward to apply the technique for any ($p$,$q$). Consider the control sequence for ($p$,$q$)=(3,5), and a similar construction of groups:
\begin{figure}[H]
\centering
	\vspace{-2pt}
	\begin{tikzpicture}

	\draw[gray, very thick] (-8.5pt,26pt) -- (-8.5pt,4pt);
	\draw[gray, very thick] (2.5pt,26pt) -- (0pt,4pt);
	\draw[gray, very thick] (14pt,26pt) -- (8pt,4pt);
	\draw[gray, very thick] (25.5pt,26pt) -- (34pt,4pt);
	\draw[gray, very thick] (35.5pt,26pt) -- (80pt,4pt);
	\draw[gray, very thick] (61.5pt,26pt) -- (40pt,4pt);
	\draw[gray, very thick] (73pt,26pt) -- (48.5pt,4pt);
	\draw[gray, very thick] (84pt,26pt) -- (84pt,4pt);
	\draw[gray, very thick] (95pt,26pt) -- (119.5pt,4pt);
	\draw[gray, very thick] (106pt,26pt) -- (128pt,4pt);
	\draw[gray, very thick] (132pt,26pt) -- (88pt,4pt);
	\draw[gray, very thick] (141.5pt,26pt) -- (134pt,4pt);
	\draw[gray, very thick] (153pt,26pt) -- (160pt,4pt);
	\draw[gray, very thick] (165pt,26pt) -- (168pt,4pt);
	\draw[gray, very thick] (176.5pt,26pt) -- (176.5pt,4pt);
	
	\draw[white, fill=white] (-14pt,24.5pt) -- (-14pt,36pt) -- (210pt,36pt) -- (210pt,24.5pt) -- cycle;
	\node[anchor=center] at (15pt,30pt) {$ A \; B \; C \; D \; E$};
	\node[anchor=center] at (84pt,30pt) {$ B \; C \; D \; A \; E$};
	\node[anchor=center] at (153pt,30pt) {$ C \; D \; B \; A \; E$};
	\node[anchor=center] at (49.5pt,27pt) {\Large .};
	\node[anchor=center] at (118.5pt,27pt) {\Large .};
	
	\draw[white, fill=white] (-14pt,-6pt) -- (-14pt,6pt) -- (182pt,6pt) -- (182pt,-6pt) -- cycle;
	\node[anchor=center] at (0pt,0pt) {$( \, 1 \; 2 \; 3 \,)$};
	\node[anchor=center] at (42pt,0pt) {$( \, 4 \; 2 \; 3 \,)$};
	\node[anchor=center] at (84pt,0pt) {$( \, 5 \; 4 \; 3 \,)$};
	\node[anchor=center] at (126pt,0pt) {$( \, 1 \; 5 \; 4 \,)$};
	\node[anchor=center] at (168pt,0pt) {$( \, 2 \; 1 \; 5 \,)$};

\end{tikzpicture}
	\vspace{-7pt}
\end{figure}
The problem in the above sequence is that by the third bracket, the number 3 has already occurred 3 times, so by the time we process this bracket, group $C$ on the upper level has to switch for the third time. Since each upper-level group becomes switchable at the same time, this means that by this point, all groups $A$, ..., $E$ now must be switchable for the third time; in particular, group $E$ too. That must mean that group $E$ has already switched at least twice previously; however, the third bracket contains the very first occurrence of number 5, so at least for one of the two switches of group $E$, the nodes labeled `5' on the lower level have wasted an opportunity to switch, so they could not switch a $\mu=\frac{5}{3}$ factor more than their upper neighbors. 

Essentially, the problem with the sequence is that the third bracket contains both the $j^{\text{th}}$ occurrence of one number and the $(j+2)^{\text{th}}$ occurrence of another (numbers 5 and 3, respectively). Because of the $(j+2)^{\text{th}}$ occurrence of a number in the bracket, all groups have to become switchable $(j+2)$ times, and hence already be switched $(j+1)$ times by the time we reach this point. However, if nodes labeled with another number are only switching at this point for the $j^{\text{th}}$ time, then one of the $(j+1)$ switches of their control group has not been used. Generally, given groups $X$ and $Y$, if there is a bracket in the sequence that contains the $j^{\text{th}}$ occurrence of the number corresponding to $X$ and the $(j+2)^{\text{th}}$ occurrence of the number corresponding to $Y$, then we say that $X$ and $Y$ are \textit{in contradiction} with each other (in the given bracket). For ($p$,$q$)=(3,5), $C$ and $E$ are in contradiction in the third bracket as discussed. For ($p$,$q$)=(2,4), we can see that there is no contradiction between any two letters.

Note that such contradictions are the only possible source of a problem; given a control sequence with no contradiction, there always exists a valid switching sequence of the upper groups. Since the control sequence itself guarantees that the occurrence numbers can never differ by more than 2, the lack of contradictions ensures that the difference between occurrences is at most 1 at any point. Hence whenever we require the $(j+1)^{\text{th}}$ switching of a specific upper group, we can simply switch all upper groups that have not been switched for the $j^{\text{th}}$ time yet; by this point, the lower neighbors of each such group have certainly been switched for the $(j-1)^{\text{th}}$ time already, so we are indeed not wasting any switches. Thus our goal is to somehow avoid contradictions in the control sequence.

Generally, devising a control gadget for any $p$ and $q$ is a challenging task. In the following, we present the technique of \textit{shifting}, which allows us to considerably increase the number of ($p$,$q$) pair for which we can devise a control gadget. We first illustrate the technique on the concrete example of ($p$,$q$)=(3,5).

\subsection{Subset shifting} \label{sec:shift}

In the above example of ($p$,$q$)=(3,5), the only problem essentially was that the second instance of $E$ always preceded the third $A$. However, the sequence $(ABCD.ABCDE.ABCDE.E)$ would, on the other hand, cause no problems at all.
 
Therefore, the key idea is that we can simply skip the very first switching of the group $E$, and only switch the groups $ABCD$ in this case. Then every further time when the upper groups become switchable, we do switch every group. Finally, when the upper groups become switchable for the fourth time, we start by switching the group $E$. At this point, the sequence of switched blocks is exactly $(ABCD.ABCDE.ABCDE.E)$, which will then again be followed by $ABCD$ when we also switch the other groups for the fourth time. A concatenation of such sequences yields a sequence where the group $E$ is effectively in a different phase, delayed from the other groups by 1 round.

Note that shifting $E$ skips an opportunity to switch group $E$ in the very first switching of the upper groups, and also an opportunity to switch $ABCD$ at the very last switching of the upper groups. Hence, if the number of switches on a given level is $s$, then with this technique, the number of switches on the next level will not be $s \cdot \frac{1-\varphi}{\lambda+\varphi}$, but only $(s-1) \cdot \frac{1-\varphi}{\lambda+\varphi} = \frac{1-\varphi}{\lambda+\varphi} \cdot s - \frac{1-\varphi}{\lambda+\varphi}$. However, one can see that this only adds up to a loss of (an arbitrarily small) $\epsilon_1$ in the exponent of the number of switches: for any $\epsilon_1>0$, we can select a constant $s_0$ high enough such that $\frac{1-\varphi}{\lambda+\varphi}  \cdot s_0 -\frac{1-\varphi}{\lambda+\varphi} > s_0 \cdot (\frac{1-\varphi}{\lambda+\varphi}-\epsilon_1)$ (note that this is very similar to the technique we used when relaxing the CPS definition; nodes that switch at most $s_0$ times are essentially considered new base nodes). Then due to this inequality, the number of switches of each group on the lowermost level of our construction is still

\[ \Omega \left( \left( \frac{1-\varphi}{\lambda+\varphi}-\epsilon_1 \right)^{\frac{1}{\log\left(\frac{1-\varphi}{\varphi}\right)} \cdot n} \right) = \Omega \left( n^{\frac{\log \left(\frac{1-\varphi}{\lambda+\varphi}-\epsilon_1\right)}{\log\left(\frac{1-\varphi}{\varphi}\right)} \cdot n} \right) = \Omega \left( n^{f(\lambda)-\epsilon_2} \right), \]
for an arbitrarily small $\epsilon_2$, as we are using $\varphi=\varphi^*(\lambda)$, and $f(\lambda)$ is continuous. Also, note that since each such loss of $\epsilon$ in the exponent can be arbitrarily small, the different such losses in the exponent can be merged into one common $\epsilon$ in the final running time.

Note that both in majority and minority, skipping the very first or very last switch of a node does not create any problems colorwise. Skipping the last switching opportunity only results in ending up with the opposite color in the final state. For each node that is supposed to skip the first switching opportunity, we have to invert its original color, such that the nodes already start with the color they would acquire if group $E$ was also switched at the first opportunity.

\subsection{Shifting in general}

Note, however, that this technique only allows us to shift a specific subset of the upper groups by 1. A crucial property of shifting is that the subsets at the beginning and the end of our modified sequence ($ABCD$ and $E$, respectively) form a disjoint partitioning of the upper neighbor groups. If we were to use the sequence $(ABCD.ABCD.ABCDE.E.E)$, then with the concatenation of such sequences, instead of skipping one switch altogether, the groups would skip a switch at every third opportunity. This would effectively reduce the number of switches on each next level to  $s \cdot \frac{1-\varphi}{\lambda+\varphi} \cdot \frac{2}{3}$, which would have a major effect on stabilization time.

This is also the reason why shifting does not provide a general solution for any ($p$,$q$) pair. Consider, for example, the control sequence for ($p$,$q$)=(7,9), which looks as follows:
\[ (1234567)(2345678)(3456789)(4567891)(5678912)(6789123)(7891234)(8912345)(9123456) \]

Here, the $3^{\text{rd}}$ bracket contains the $1^{\text{st}}$ occurrence of 9 and the $3^{\text{rd}}$ occurrence of 3, while the $6^{\text{th}}$ bracket contains the $4^{\text{th}}$ occurrence of 3 and the $6^{\text{th}}$ occurrence of 7. This implies that for a correct solution, the upper neighbors of 9 (i.e., group $I$) should be shifted at least 1 further than the upper neighbors of 3 (group $C$), and the upper neighbors of 3 (group $C$) shifted at least one further than the upper neighbors of 7 (group $G$). However, then group $I$ is shifted at least 2 steps away from group $G$ (i.e., must skip at least 2 initial rounds to be sufficiently later than $G$), which, as discussed above, is not viable.

The main goal of shifting is to separate the groups that are in contradiction with each other in a specific bracket. We say that a subset of letters (i.e., groups) is \textit{consistent} if there is no two groups of the subset are in contradiction in any bracket. In general, shifting provides a solution for a ($p$,$q$) pair if the letters can be partitioned into two consistent subsets. We call these two subsets \textit{blocks}, and we also refer to the partitioning as consistent if both of its blocks are consistent. For $(p,q)=(3,5)$, a partitioning is consistent exactly if it places $A$ and $E$ in different blocks.

It depends on the concrete value of $p$ and $q$ whether a consistent partitioning (into two groups) exists, i.e., whether the shifting technique provides a valid control gadget. In the following section, we show that such a partitioning always exists if $\mu \leq \frac{3}{5}$, that is, for $\lambda$ less than approximately $0.476$.

\begin{lemma} \label{lem:upto35}
Under Rule II with $\lambda < 0.476$, for any $\epsilon > 0$, there exists a graph construction and initial coloring where majority/minority processes stabilize in time $\Omega(n^{1+f(\lambda)-\epsilon})$.
\end{lemma}

While these $\mu \leq \frac{3}{5}$ values allow a relatively simple proof of consistency, these are not the only $\mu$ values for which shifting provides a valid solution. For larger $\mu$, however, the existence of a consistent partitioning depends on multiple factors, including how large the integers $p$ and $q$ are. For example, the case $(p,q)=(5,7)$ can also be partitioned consistently, and thus the shifting technique provides a valid construction for $\mu=\frac{5}{7}$. This corresponds to $\lambda \approx 0.635$, which is a notably larger value than $0.476$.

\begin{lemma}
Under Rule II with $\lambda \approx 0.635$, for any $\epsilon > 0$, there exists a graph construction and initial coloring where majority/minority processes stabilize in time $\Omega(n^{1+f(\lambda)-\epsilon})$.
\end{lemma}

Thus in general, the concept of levels allows us to devise a construction idea to prove the lower bound for any $\lambda$ value. However, to obtain an actual realization of such a construction for every $\lambda \in (0,1)$, it remains to solve the combinatorial task of forming a control gadget for the remaining $\lambda$ values that are not covered by the shifting method.

\subsection{Consistent partitioning for $\mu \leq \frac{3}{5}$}

We now discuss how to partition the upper groups into two consistent groups for any $\mu \leq \frac{3}{5}$. Note that while our method shifts a block of groups on the upper level (e.g. groups $A$ and $B$), the consistency of this block depends on the groups' lower neighbors (e.g., where nodes labeled 1 and 2 appear in the control sequence below). Thus, for simplicity, we refer to each group not by its letter, but by the number assigned to its neighbors on the level below, and our goal is to find a consistent partitioning of the numbers $(1, ..., q)$ into 2 blocks.

Recall that $b=\frac{q-p}{2}$, i.e. the number of different elements in two consecutive brackets of the control sequence. For now, let us first assume that $p \geq 2b$.

Furthermore, let us use $B_{\ell}$ to denote a block formed from any $\ell$ consecutive numbers in $(1, ..., p)$, i.e. containing (the letters for) the numbers $i+1, i+2, ..., i+{\ell}$ for some $0 \leq i \leq p - \ell$. Also, let $B'_b$ and $B''_b$ denote the blocks formed from the numbers $(p+1, ..., p+b)$ and $(p+b+1, ..., q)$, respectively; note that these both consist of $b$ numbers indeed.

\begin{lemma} \label{lem:Bcons}

Any block $B_{2b}$ is consistent.

\end{lemma}

\begin{proof}

Note that a control sequence is developed as follows: there is a starting point $h_s$ and an endpoint $h_e$, which are shifted in each step in a modular fashion (i.e., $q$ is followed by 1 again). Initially, $h_s$ and $h_e$ are at $1$ and $p+1$, respectively, so the first bracket of the control sequence contains the numbers $[h_s, h_e)$. In each step, both points are shifted further ahead by $b$ (modulo $q$). Since $h_e$ starts at $p+1$, after two steps, $h_e$ will arrive at 1, and then follow the same pattern from here as $h_s$ from the beginning. Hence, the position of $h_e$ in the $j^{\text{th}}$ step is always the same as the position of $h_s$ in the $(j-2)^{\text{th}}$ step.

The initial bracket of the sequence contains all elements of $B_{2b}$. After some steps, we have $h_s>i+1$ (for the first number $i+1$ in the group); let $h_s^1$ denote the value of $h_s$ in this step. This shows that in this step, only the numbers ($h_s^1$, ..., $i+2b$) will be present in the next bracket. Then in the following step, $h_s$ falls within the range of $B_{2b}$ again, so only the numbers ($h_s^1+b$, ..., $i+2b$) will be contained in the next bracket. The key observation is that in the step after this, $h_e$ will be equal to $h_s^1$ (it always takes the same position as $h_s$ did two rounds ago), hence the next bracket will contain the groups ($i+1$, ..., $h_s^1-1$) of $B_{2b}$, which is exactly the complement of groups two rounds ago. Similarly, the bracket of the next step contains ($i+1$, ..., $h_s^1+b-1$), the complement of the bracket from two steps before. After this point, each element of $B$ will have occurred the same number of times again.

Therefore, whenever we have brackets that only contain a subset of $B_{2b}$, they are always organized as follows. Before this point, each group in $B_{2b}$ has the same occurrence number. Then the following two brackets contain some subsets $S_1$ and $S_2$ of $B_{2b}$, and after this, the next two brackets contain exactly the complements of $S_1$ and $S_2$. This pattern ensures that regardless of the content of $S_1$ and $S_2$, no bracket has a difference of 2 in occurrence numbers, and after the pattern, all groups have the same occurrence numbers again.

It is worth pointing out that this heavily relies on the fact that the size of $B_{2b}$ is at most $2b$, and hence whenever $h_s$ or $h_e$ falls within the range of $B_{2b}$, it is guaranteed that it already surpasses the entire range of $B_{2b}$ in the second step after this. For example, in case of $(p,q)=(7,9)$ shown above, the block $(3,4,5,6,7)$ does not obey this property, since the starting point falls into it in 4 consecutive rounds, and hence it is not consistent. \qedhere

\end{proof}

Note that the same proof holds for any continuous block $B$ within $(1, ..., p)$ if it has size at most $2b$. Specifically, for the case of $p < 2b$, putting all of $(1, ..., p)$ together still forms a consistent block.

\begin{lemma} \label{lem:B12cons}

Blocks $B'_b$ and $B''_b$ are both consistent.  

\end{lemma}

\begin{proof}

Blocks $B'_b$ and $B''_b$ follow the same behavior as any block $B_b$ described in Lemma \ref{lem:Bcons}, except for not being included in the first 1 and first 2 brackets, respectively. Hence, the same reasoning shows that these blocks are also consistent. \qedhere

\end{proof}

It remains to show that we can merge the blocks $B'_b$ and $B''_b$ with the blocks in $(1, ..., p)$ to obtain a consistent partitioning into two blocks for smaller $\mu$ values. For this, we introduce some new notation. Let us denote the block corresponding to numbers $(1, ..., b)$ by $B_b^{\text{first}}$, and the block corresponding to numbers $(p-2b+1, ..., p)$ by $B_{2b}^{\text{last}}$.

\begin{lemma} \label{lem:Mergecons1}

The block $B_{2b}^{\text{last}} \cup B'_b$ is consistent.

\end{lemma}

\begin{proof}

Our previous lemmas show that both $B_{2b}^{\text{last}}$ and $B'_b$ are consistent separately. Together, they form a block of $3b$ consecutive numbers. Note that the only reason why the proof of Lemma \ref{lem:Bcons} does not apply to blocks of length $3b$ is that $h_s$ can fall within the range of the block on 3 consecutive occasions, and thus a bracket could simultaneously have the $(j+2)^{\text{th}}$ occurrence of the last few numbers and the $j^{\text{th}}$ occurrence of the first few numbers. However, in our case, $B'_b$ is not contained in the first bracket ($h_e=p+1$ initially), so the occurrence number of all nodes in $B'_b$ is always smaller by 1 than the same occurrence numbers in the $B_{3b}$ case. Hence even if $h_s$ falls into the range of the block 3 consecutive times, the resulting bracket only contains the $(j+1)^{\text{th}}$ occurrence of the last nodes in $B'_b$, and the $j^{\text{th}}$ occurrence of the first nodes in $B_{2b}^{\text{last}}$. \qedhere

\end{proof}

\begin{lemma} \label{lem:Mergecons2}

The block $B_{b}^{\text{first}} \cup B''_b$ is consistent.

\end{lemma}

\begin{proof}

The first bracket of the control sequence contains all elements of $B_{b}^{\text{first}}$. The second bracket contains none of the numbers in the merged block, while the third bracket only contains the elements of $B''_b$. Up to this point, all elements of the merged block appear exactly once. From here, the merged block simply behaves as any block $B_{2b}$ in the proof of Lemma \ref{lem:Bcons}: it is a block of $2b$ consecutive number, such that each have the same occurrence number in the beginning. \qedhere

\end{proof}

Note that this already provides a construction proving Lemma \ref{lem:upto35}. If $\mu \leq \frac{3}{5}$, then $p \leq 3b$, so $B_b^{\text{first}}$ and $B_{2b}^{\text{last}}$ together already cover all numbers in $(1, ..., p)$. Thus the merged blocks in Lemmas \ref{lem:Mergecons1} and \ref{lem:Mergecons2} cover all upper groups, giving a consistent partitioning. Therefore, the shifting technique provides a valid control gadget if we shift all the upper groups in $B_{b}^{\text{first}} \cup B''_b$ by 1.

On the other hand, for general $(p,q)$ pairs with $\mu > \frac{3}{5}$, the groups corresponding to $(1, ..., p)$ can not necessarily be partitioned into two consistent blocks, and thus we cannot obtain a valid control gadget with the shifting method, as in the example of $(p,q)=(7,9)$ before.

Note that some of the above statements would have to be slightly reformulated to also hold for very small $\mu$ values, when even $p < b$. However, for such small $\mu$, the control sequence is always guaranteed to be contradiction-free, so the shifting technique is not even required to form a control gadget.

\subsection{An easier lower bound}

We also briefly note that a simple technique allows us to show a slightly weaker lower bound in case of any $\lambda$, even without the shifting technique. Recall that the idea of upper groups (i.e., assigning a letter and a number to a node) allowed us to handle any case where the occurrence numbers in any bracket of a control sequence differ by at most 1. Note that in a control sequence, the occurrence numbers in any bracket can differ by at most 2 in any case, so increasing this limit by 1 more would already provide a control gadget for any $\lambda$.

\begin{figure}
\centering
\captionsetup{justification=centering}
	\includegraphics[width=0.6\textwidth]{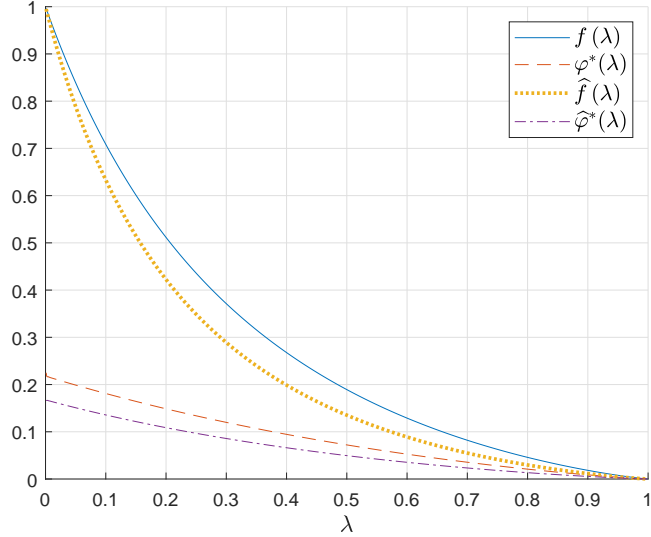}
	\caption{Plot of $\widehat{f}(\lambda)$ and $\widehat{\varphi}^*(\lambda)$, besides $f(\lambda)$ and $\varphi^*(\lambda)$}
	\label{fig:func_loose}
\end{figure}

Consider the idea of placing a level of \textit{relay nodes} between any two consecutive levels of our construction, taking a mediator role between the two levels. While previously, the nodes labeled $A$ in the upper level were connected to the nodes labeled $1$ in the lower level, we now remove these edges, an instead connect all these nodes to a set of relay nodes $R_{A/1}$ inbetween. This extra level then allows us to temporarily store conflicts, and relay them to the lower level in a timing of our choice, which is already enough to implement the control sequence for any $\lambda$.

The drawback of the technique, however, is that the relay nodes now also waste conflicts. While previously both the downdegree of the upper level and the updegree of the lower level was $d$, now in order to allow the relay nodes to be dominated by their upper neighbors, we now must select the downdegree of the upper level and the updegree of $R_{A/1}$ to be $d$, and then the downdegree of $R_{A/1}$ and the updegree of the lower level to be $\frac{1-\lambda}{1+\lambda} \cdot d$. In practice, this means that every new level of the construction will imply an extra degree decrease factor of $\frac{1-\lambda}{1+\lambda}$.

For every new level, the number of edges now decreases by $\frac{\varphi}{1-\varphi} \cdot\frac{1-\lambda}{1+\lambda}$, so the optimal choice of $\varphi$ also changes. Hence this construction requires a new choice $\widehat{\varphi}^*$ of output rate, which will then, analogously to the original case, result in a stabilization time defined by the function

\[ \widehat{f}(\lambda) := \max_{\varphi \in (0,\frac{1-\lambda}{2}]} \; \frac{\log\left( \frac{1-\varphi}{\lambda+\varphi} \right)}{\log\left( \frac{1-\varphi}{\varphi} \cdot \frac{1+\lambda}{1-\lambda} \right)}. \]

This alternative lower bound function is shown in Figure \ref{fig:func_loose}. While this lower bound does leave some gap to the upper bound of $O(n^{1+f(\lambda)+\epsilon})$, it has the advantage of being easy to show for any $\lambda$, without having to devise complicated control gadgets.

\begin{theorem}
Under Rule II with any $\lambda \in (0,1)$, for any $\epsilon > 0$, there exists a graph construction and initial coloring where majority/minority processes stabilize in time $\Omega(n^{1+\widehat{f}(\lambda)-\epsilon})$.
\end{theorem}

\subsection{Above the uppermost level}

Furthermore, the uppermost level of the construction needs to be discussed separately, since in order to make the construction behave as we described, we also have to ensure that the nodes of the uppermost level already execute the control sequence a constant $s_0$ number of times.

The reason why this is necessary is that on each level of the construction, we lose a constant number of switches due to two different factors. On the one hand, recall that if we apply the subset shifting method, then this leaves exactly 1 switch of each node on each level unused. On the other hand, if each node in the given level switches $s$ times, the next level cannot always switch $s \cdot \frac{1-\varphi}{\lambda+\varphi}$ times if this expression is not an integer. In fact, if each node switches $t$ times in the control sequence of our control gadget (with $t=O(1)$), this allows for only $\left \lfloor{\frac{s}{t}} \right \rfloor$ complete executions of the control sequence on the upper level, and hence only
\[ \left \lfloor{\frac{\left \lfloor{\frac{s}{t}} \right \rfloor \cdot \frac{1-\varphi}{\lambda+\varphi}}{t}} \right \rfloor \]
complete executions of the control sequence on the lower level. Thus due to these two factors, the number of switches does not increase from $s$ to $s \cdot \frac{1-\varphi}{\lambda+\varphi}$ for each new level, but only to $s \cdot \frac{1-\varphi}{\lambda+\varphi} - O(1)$ for some constant.

As discussed already in Section \ref{sec:shift}, we can overcome this by ensuring that the nodes of each level switch at least $s_0$ times for a specific constant $s_0$, at the cost of losing a factor $\epsilon$ from the exponent of our lower bound. The smaller the $\epsilon$ loss we tolerate, the larger the minimal switches $s_0$ we have to ensure for each (i.e., even the uppermost) level.

There is a simple method to ensure that each node in the uppermost level of the construction switches $s_0$ times, for any constant $s_0$. A similar technique was already used in the weighted constructions of \cite{minorityW}. Since our control gadgets have constant size, there are at most constantly many different `type of' nodes on the uppermost level. For all these sets $V_0$ of uppermost level nodes (that have the same role in different control gadgets), we can connect $V_0$ to a group $V'_0$ on an even higher pseudo-level, such that each edge between $V_0$ and $V'_0$ has a conflict initially. If nodes in $V_0$ have a downdegree of $d$, then we connect each node in $V_0$ to $\frac{\lambda+1}{\lambda-1} \cdot d$ nodes in $V'_0$. This ensures that each node in $V_0$ is switchable initially, while the extra nodes in $V'_0$ and extra edges to $V'_0$ still remain in the magnitude of $|V_0|$ and $|V_0| \cdot d$, respectively.

We can then continue this in a similar fashion, and add another group $V''_0$ above $V'_0$, connected with even more edges, in order to make $V'_0$ initially switchable. After adding $s_0$ such pseudo-levels above, and then unfolding them from bottom to top (i.e., first switching $V_0$, then $V'_0$ and then $V_0$, then $V''_0$ and $V'_0$ and then $V_0$, and so on), we obtain a way to switch the nodes of $V_0$ altogether $s_0$ times, at a timing of our choice. Since $s_0$ is a constant, executing this process for a specific $V_0$ does not change the magnitude of nodes or edges in the graph. As our control gadgets consist of constantly many nodes, adding distinct such pseudo-levels for all the constantly many $V_0$ sets still does not affect the magnitude of the nodes and edges.

\subsection{Divisibility challenges}

Besides the difficulty of devising a control gadget for every $\lambda$, there is another problem to address in the construction.

Assume that the input-output rate $\frac{1-\varphi}{\varphi}$ can be expressed as (or, in the irrational case, approximated by) a rational number $\frac{p'}{q'}$ with $p',q' \in \mathbb{Z}$ (note that this $p'$ and $q'$ has no relation to our choice of $p$ and $q$, which are used to approximate $\mu$).

This means that if a node in a specific level has downdegree $d$, then it has to have updegree $\frac{p'}{q'} \cdot d$ for the optimal rate $\varphi^*(\lambda)$. However, in our construction, that would imply that the level above has updegree $\left( \frac{p'}{q'} \right)^2 \cdot d$, the following level $\left( \frac{p'}{q'} \right)^3 \cdot d$, and so on. In order for all of these numbers to be integers, $d$ would have to be divisible by $q'$ many times ($\Theta(\log n)$ times). This is clearly not possible, especially for the lowermost levels, where $d$ is a constant.

We can overcome this problem by slightly modifying the number of nodes (i.e., the number of control gadgets) on each level. Let us select $k \in \mathbb{Z}$ such that $\frac{p'}{q'} \in \left[k, k+1 \right)$ holds (note that $\varphi^*(\lambda)<0.22$ for any $\lambda$, and thus $\frac{1-\varphi}{\varphi}>3$ in any case). Assume we have a specific level where each node has an updegree of $d$. If the level above had the same number of nodes, than that would imply a downdegree of $d$ for each node above, and consequently, an updegree of $\frac{p'}{q'} \cdot d$. However, instead, we can increase the size of the level above by a factor of $\frac{p'}{k \cdot q'}$, resulting in a downdegree of only $\frac{k \cdot q'}{p'} \cdot d$, and thus an updegree of $\frac{k \cdot q'}{p'} \cdot \frac{p'}{q'} \cdot d = k \cdot d$ on the level above. Similarly, if we decrease the size of the next level by a factor of $\frac{p'}{(k+1) \cdot q'}$, then the next updegree $ (k+1) \cdot d$ will similarly be an integer.

The general idea is to follow this technique to ensure that the degree remains an integer after each such level. Note, however, that in order not to change the construction significantly, we need to select a combination of $k$-s and $(k+1)$-s such that their product over all $L$ levels is relatively close to $\left( \frac{p'}{q'} \right)^L$. In case of too many $k$-s, the uppermost level would be significantly larger than the lowermost one, not giving us enough frequently-switching nodes on lower levels. In case of too many $(k+1)$-s, the degree of nodes would grow significantly faster than $\frac{p'}{q'}$ on a level, resulting in less than $L$ levels altogether (since the degree on the uppermost level would have to be larger than $\Theta(n)$). A possible solution is to select the largest combination of $k$-s and $(k+1)$ that is still below $\left( \frac{p'}{q'} \right)^L$, which is therefore at least $\frac{k}{k+1} \cdot \left( \frac{p'}{q'} \right)^L$. This ensures that there is only at most a constant variance in level sizes, and that the uppermost level has degree which is only a constant factor lower than it would be with $\left( \frac{p'}{q'} \right)^L$.

Note that our divisibility solution itself raises another minor divisibility problem: changing the size of specific levels by a factor of $\frac{p'}{k \cdot q'}$ or $\frac{p'}{(k+1) \cdot q'}$ might also mean that the following level should have a non-integer number of control gadgets. However, we can easily overcome this. For simplicity, let us analyze the process in the other direction, from uppermost to lowermost level. Whenever the level size change by the given factor would result in a non-integer number of control gadgets, we can simply round this number down, and connect the few extra edges to a dummy gadget on the level below that we do not use. With possibly one less actual control gadget, the number of nodes can only decrease by a constant on each new level, hence we only lose $O(\log(n))$ nodes by the lowermost level. Since each level consists of $\widetilde{\Theta}(n)$ nodes, this does not affect the magnitude of nodes on any level.

\section{Discussion of $f(\lambda)$} \label{App:C}

We now discuss the functions $f(\lambda)$ and $\varphi^*(\lambda)$ in more detail. The diagram of both functions have already been presented in the main part of the paper. This shows that both functions are continuous and monotonously decreasing. The function $f(\lambda)$ takes values in $[0,1]$, while $\varphi^*(\lambda)$ takes values between 0 and approximately 0.2178.

Let us introduce the notation
\[ g(\lambda, \varphi) = \frac{\log\left( \frac{1-\varphi}{\lambda+\varphi} \right)}{\log\left( \frac{1-\varphi}{\varphi} \right)}. \]
In order to find the optimal $\varphi$, one would have to differentiate $g(\lambda, \varphi)$:
\[ g'_\varphi (\lambda, \varphi) = \frac{(\lambda+1) \cdot \varphi \cdot \log (\frac{1-\varphi}{\varphi}) - (\lambda + \varphi) \cdot \log (\frac{1-\varphi}{\lambda+\varphi})}{(\varphi-1) \cdot \varphi \cdot (\lambda+\varphi) \cdot \log^2 (\frac{1-\varphi}{\varphi})}. \]
Thus at a local minimum, we have
\[ (\lambda+1) \cdot \varphi \cdot \log \left( \frac{1-\varphi}{\varphi} \right) = (\lambda + \varphi) \cdot \log \left( \frac{1-\varphi}{\lambda+\varphi} \right) .\]
In order to obtain $\varphi^*(\lambda)$, we would have to solve this for $\varphi$, with $\lambda$ as a parameter. To our knowledge, there is no closed-form solution to this problem.

Note that if we split the logarithms into subtractions, we also obtain an alternative formulation of this equation.
\[ (\lambda + \varphi) \cdot \log ( \lambda+\varphi ) = (\lambda+1) \cdot \varphi \cdot \log ( \varphi) + \lambda \cdot (1-\varphi) \cdot \log (1-\varphi) . \]

\subsection{Lookup table of function values}

Finally, we show the approximate values of $f(\lambda)$ and $\varphi^*(\lambda)$ for a wide range of $\lambda$ values between 0 and 1. Besides, we also show the input switching rate $\mu=\frac{\lambda+\varphi^*(\lambda)}{1-\varphi^*(\lambda)}$ for these $\lambda$ values. The values are illustrated in Table \ref{tab:functions}.

\setlength\tabcolsep{4pt}
\begin{table}[H]
\captionsetup{justification=centering}
\vspace{24pt}
\centering
  \begin{tabular}{ c || c | c | c | }
    \hline
		$\lambda$ & $f(\lambda)$ & $\varphi^*(\lambda)$ & $\mu(\lambda)$ \\ \hline \hline
    0.05 & 0.839 & 0.199 & 0.311 \\ \hline
		0.10 & 0.709 & 0.181 & 0.343 \\ \hline
		0.15 & 0.601 & 0.164 & 0.376 \\ \hline
		0.20 & 0.512 & 0.149 & 0.410 \\ \hline
		0.25 & 0.436 & 0.134 & 0.443 \\ \hline
		0.30 & 0.371 & 0.120 & 0.477 \\ \hline
		0.35 & 0.316 & 0.107 & 0.512 \\ \hline
		0.40 & 0.268 & 0.095 & 0.546 \\ \hline
		0.45 & 0.226 & 0.083 & 0.581 \\ \hline
		0.50 & 0.189 & 0.072 & 0.617 \\ \hline
		0.55 & 0.157 & 0.062 & 0.653 \\ \hline
		0.60 & 0.129 & 0.053 & 0.689 \\ \hline
		0.65 & 0.104 & 0.044 & 0.726 \\ \hline
		0.70 & 0.082 & 0.036 & 0.763 \\ \hline
		0.75 & 0.063 & 0.028 & 0.800 \\ \hline
		0.80 & 0.046 & 0.021 & 0.838 \\ \hline
		0.85 & 0.031 & 0.015 & 0.877 \\ \hline
		0.90 & 0.018 & 0.009 & 0.917 \\ \hline
		0.95 & 0.008 & 0.004 & 0.958 \\ \hline
  \end{tabular}
  \vspace{4pt}
    \caption{Values of our functions for some specific $\lambda$ parameters.}
	\label{tab:functions}
\end{table}

\end{appendices}

\newpage

\end{document}